\documentclass[sigconf]{acmart}

\RequirePackage{OpenProc-Mods}

\AddToHook{cmd/maketitle/before}{%

\def\EDBTISSN{2367-2005}
\def\EDBTISBN{978-3-89318-106-3}
\setcopyright{none}
\copyrightyear{\copyright\ 2026 Copyright held by the owner/author(s).
  Published on \url{OpenProceedings.org} under ISBN \EDBTISBN, series ISSN \EDBTISSN.
  Distribution of this paper is permitted under the terms of the
  Creative Commons license CC-by-nc-nd 4.0}
\acmYear{2026}
\acmDOI{}
\acmISBN{}

\acmConference[EDBT '27]{Extending Database Technology}{6-9 April 2027}{Lille (France)}

\settopmatter{printacmref=false, printccs=false, printfolios=false}

\newsavebox{\ximagebox}
\newlength{\ximageheight}
\newsavebox{\xglyphbox}
\newlength{\xglyphheight}
\newcommand{\xbox}[1]%
  {\savebox{\ximagebox}{#1}%
  \settoheight{\ximageheight}{\usebox{\ximagebox}}%
  \savebox{\xglyphbox}{\color{white}\char32}%
  \settoheight{\xglyphheight}{\usebox{\xglyphbox}}%
  \raisebox{\ximageheight}[0pt][0pt]{\raisebox{-\xglyphheight}[0pt][0pt]{%
    \makebox[0pt][l]{\usebox{\xglyphbox}}}}%
    \usebox{\ximagebox}%
    \raisebox{0pt}[0pt][0pt]{\makebox[0pt][r]{\usebox{\xglyphbox}}}}

\newsavebox{\LogoBox}
\sbox{\LogoBox}{\includegraphics[height=1cm]{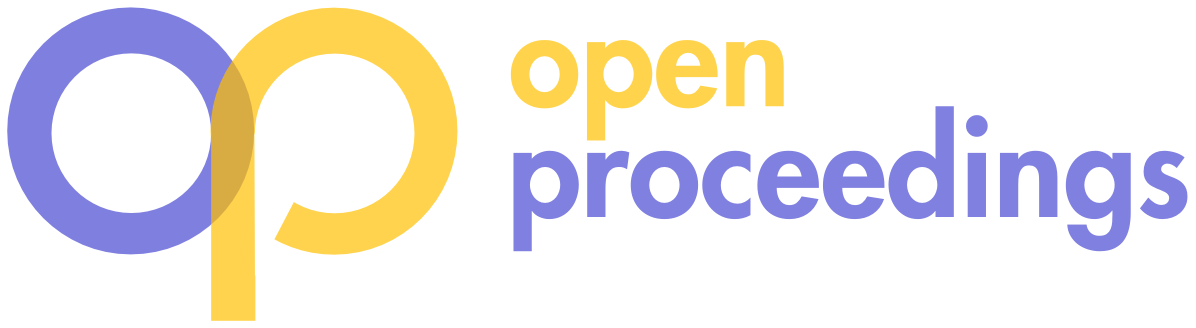}}

\fancypagestyle{firstpagestyle}{%
  \fancyhf{}%
  \fancyfoot{}%
  \fancyhead[L,C]{}
  \fancyhead[R]{\href{https://OpenProceedings.org/}{\raisebox{15pt}[0pt][0pt]{\xbox{\usebox{\LogoBox}}}}}
  }
  
\hypersetup{%
  pdfcopyright={CC-by-nc-nd},
  pdfpublisher={OpenProceedings.org -- University of Konstanz, Germany},
  pdfcreator={LaTeX with acmart, hyperref, and EDBT modifications},
  pdfvolumenum={30},
  pdfissuenum={1},
  pdfisbn={\EDBTISBN},
  pdfissn={\EDBTISSN},
  pdfdoi={10.48786/edbt.2027.xx}
}
}

\usepackage{amsmath}
\usepackage{booktabs}
\usepackage{xcolor,colortbl}
\definecolor{moemarket}{RGB}{225,245,225} 
\usepackage{multirow}
\usepackage{subcaption}
\usepackage{graphicx}
\usepackage{placeins}
\usepackage{float}

\def\BibTeX{{\rm B\kern-.05em{\sc i\kern-.025em b}\kern-.08em
    T\kern-.1667em\lower.7ex\hbox{E}\kern-.125emX}}

\begin{document}


\OPtitle{Cost-Aware Mixture-of-Experts Coordination for \\ Model Markets}




\OPauthor[0009-0005-8566-6585]{Yizhou Ma}
\affiliation{%
  \institution{University of Southampton}
  \city{Southampton}
  \country{UK}}
\email{ym7u22@soton.ac.uk}

\OPauthor[0009-0002-3937-0124]{Wenbo Wu}
\affiliation{%
  \institution{University of Southampton}
  \city{Southampton}
  \country{UK}}
\email{wenbo.wu@soton.ac.uk}

\OPauthor[0000-0002-5517-0301]{Xikun Jiang}
\affiliation{%
  \institution{Aalborg University}
  \city{Copenhagen}
  \country{Denmark}}
\email{xji@cs.aau.dk}

\OPauthor[0009-0007-8707-9370]{Zhuoqin Yang}
\affiliation{%
  \institution{University of Nottingham}
  \city{Ningbo}
  \country{China}}
\email{scxzy5@nottingham.edu.cn}

\OPauthor[0000-0001-6993-0001]{Luis-Daniel Ibáñez}
\affiliation{%
  \institution{University of Southampton}
  \city{Southampton}
  \country{UK}}
\email{l.d.ibanez@soton.ac.uk}



\begin{OPabstract}
Existing model marketplaces typically trade and select individual models as indivisible units, limiting their ability to exploit complementarities among heterogeneous experts.
This paper proposes an MoE-based model market framework that lifts Mixture-of-Experts from a model-level learning architecture to a market-level coordination mechanism.
In this framework, brokers use gating networks to coordinate multiple heterogeneous experts and deliver a composite model service. 
We formalize the market participants, service workflow, expert cost structure, and a welfare objective that combines predictive utility with heterogeneous execution costs. We then derive a cost-aware gating mechanism and market-aware training objective, and introduce a cost-adjusted revenue allocation rule that distributes residual revenue according to realized expert participation and execution cost. 
We also establish basic theoretical properties of the allocation rule, including budget balance, participation monotonicity, and cost sensitivity. 
Experiments over five random seeds on fifteen tabular and image benchmarks use independently trained and frozen neural and tree-based experts together with latency-derived execution costs. MoE Market achieves the highest mean welfare on all fifteen datasets and a lower mean expected cost than Standard MoE in every case, while maintaining competitive predictive performance. The allocation experiments further demonstrate systematic sensitivity to expert participation and cost, together with substantially lower computational overhead than exact Shapley allocation. These results suggest that MoE can serve as a market-level coordination principle for collaborative, cost-aware, and economically grounded model marketplaces.
\end{OPabstract}


\maketitle

\section{Introduction}

Model marketplaces extend data-market infrastructures by allowing buyers to acquire predictive services from model providers rather than raw datasets alone. Early data marketplaces \cite{zhang2023survey} focused on the exchange of raw datasets, whereas recent markets increasingly capture value through models trained on, or supported by, such data. In this setting, model marketplaces connect data owners, model providers, and model buyers: buyers request model services, model providers supply trained models, and data owners provide the data needed to train or improve those models.

\begin{figure}
    \centering
    \includegraphics[width=\linewidth]{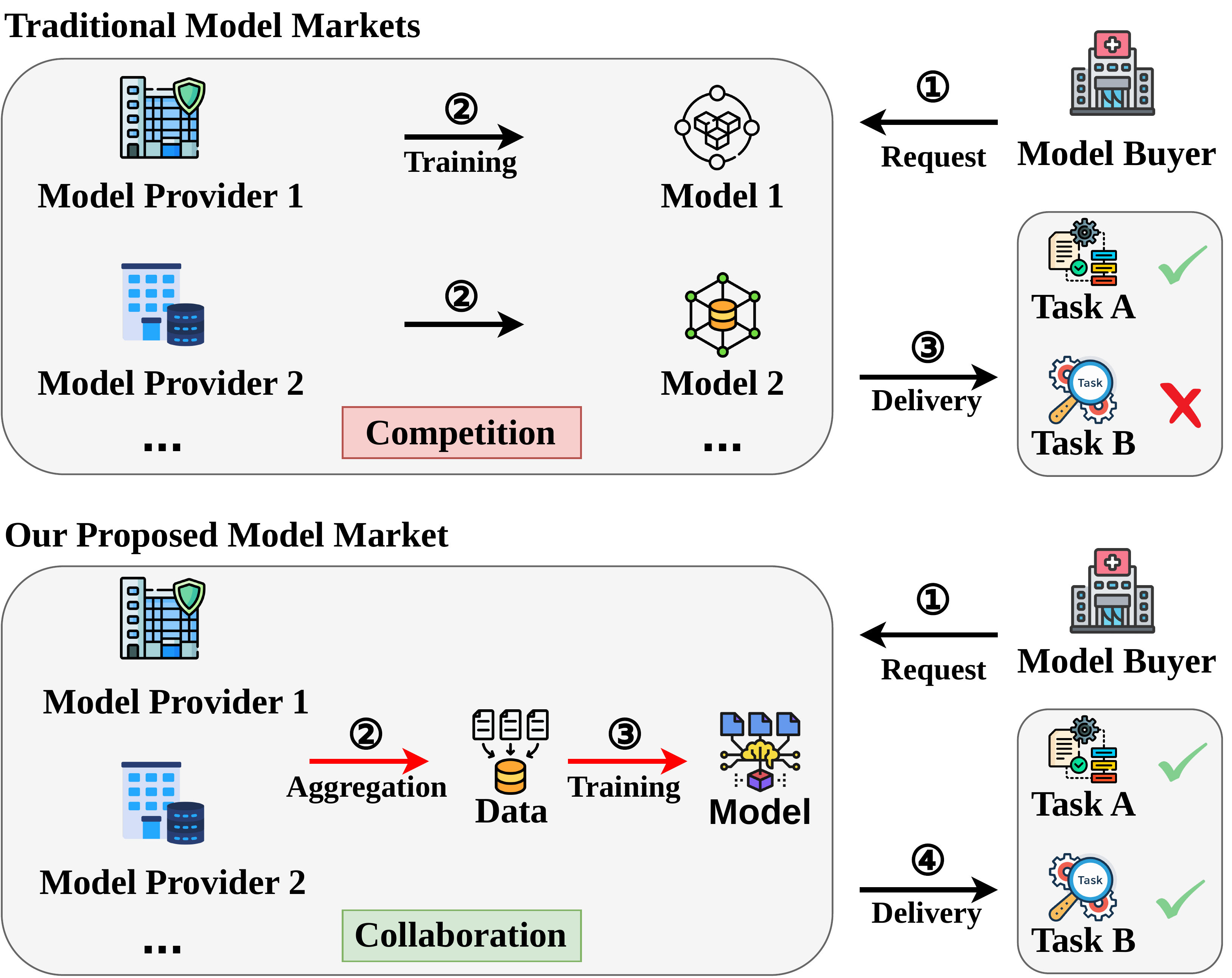}
    \caption{Competitive and cooperative model markets.
In a competitive model market, providers compete and the buyer selects
a single model. In the proposed cooperative market, the broker coordinates
multiple heterogeneous experts to deliver a composite model service.}
    \label{fig:competitive_collaborative_market}
\end{figure}

Existing model marketplaces \cite{chen2019demonstration,agarwal2019marketplace,chen2019towards,azcoitia2022survey,jiang2022pricing,jiang2023divine} typically allow buyers to buy a model from a single model provider. Multiple providers compete to satisfy buyer-specified tasks, for buyers to make their choice statically, by ranking, or versioning. This design is appropriate for competitive model trading, but it provides limited support for scenarios where the best model for the buyers' task is the result of combining the complementary capabilities of multiple providers.
For example, the developer of a fraud-detection service may benefit from coordinating an expert model specialized in user behavioral patterns with another expert model specialized in transaction anomalies. Similarly, medical AI services may benefit from coordinating models trained for different imaging modalities or clinical signals. From the providers' perspective, the lack of  cooperation and incentive mechanisms limits their ability to jointly create higher-value composite model services. These observations motivate the development of cooperative market mechanisms to incentivise and coordinate heterogeneous providers to deliver composite models. Fig.~\ref{fig:competitive_collaborative_market} illustrates the comparison between competitive and cooperative model market models.

The Mixture-of-Experts (MoE) architecture \cite{yuksel2012twenty} provides a natural mechanism for coordinating heterogeneous capabilities across tasks or input regions. A gating network dynamically assigns weights to specialized experts, allowing the system to exploit task-specific differences through selective activation and produce composite predictions that are more accurate or efficient than relying on a single expert. Most existing MoE research treats experts and gates as components of a single learning system, with objectives defined primarily at the model level. In this paper, we propose the use of MoE as a market-level coordination mechanism among heterogeneous model providers. In this way, expert routing determines not only which experts are most useful for predicting the output for a given input, but also which experts are economically worthwhile to invoke under execution costs and compensation considerations, because each expert may differ in capability, execution cost, and expected compensation. Strong foundation-model services do not eliminate this coordination problem, they can be incorporated as high-capacity experts alongside smaller specialized models.

Our previous work~\cite{ma2025mixture} introduced an initial conceptual design for applying MoE to cooperative model marketplaces. It proposed a MoE-based model market mechanism that lifts MoE from an internal component of a single model to a market-level coordination framework. However, the previous framework remains largely conceptual. In particular, it does not formalize the market objective that a broker should optimize, such as how to trade off the predictive utility of a composite service against the execution costs of the participating experts. It also does not provide a trainable cost-aware coordination mechanism or empirical evidence on the resulting allocation behavior. To address these gaps, this paper advances the MoE-based model market framework from a conceptual design to a formal and empirically grounded setting. In particular, we study the following research questions:

\begin{itemize}
    \item How can market welfare be defined so that expert coordination accounts for both model utility and heterogeneous execution costs?

    \item How can a cost-aware gating mechanism be designed and trained to coordinate heterogeneous experts while balancing model quality and service cost?

    \item How should revenue be allocated among contributing experts according to both their contribution and heterogeneous operational costs?
\end{itemize}

Our main contributions are summarized as follows:

\begin{itemize}
    \item  We formalize an \textbf{MoE-based model market} as a market-level coordination framework involving buyers, brokers, data providers, a gating network, and heterogeneous experts. Based on this formulation, we define the market workflow, expert cost model, utility function, and welfare objective for collaborative composite model service delivery.

    \item  We design a \textbf{cost-aware gating mechanism} together with a corresponding \textbf{market-aware training objective} to coordinate heterogeneous experts under explicit model quality--cost trade-offs.

    \item  We introduce a \textbf{cost-adjusted revenue allocation mechanism} for distributing the residual revenue pool among experts according to both their realized contribution and execution costs.

    \item  We evaluate the proposed framework on tabular and image benchmarks, showing that cost-aware coordination improves welfare while maintaining competitive predictive performance. We further compare the allocation rule with participation-only and exact Shapley allocation, demonstrating cost-sensitive payments and substantially lower allocation runtime.
\end{itemize}

\section{Related Work}

\textbf{Data Marketplaces.}
Data marketplaces enable the commoditization and exchange of data among data owners, brokers, and consumers, while balancing economic value with privacy considerations \cite{zhang2023survey}. 
They are commonly organized as sell-side markets, where data owners offer datasets to potential buyers; buy-side markets, where buyers post data demands; or two-sided markets, where a broker or platform matches data supply with data demand.

\noindent
\textbf{Model Marketplaces.}
Motivated by the observation that much of the value of data lies in its use to train machine learning models, \emph{model marketplaces} have emerged to directly trade machine learning models.
Chen \textit{et al.} \cite{chen2019towards} introduced the Model-Based Pricing (MBP) framework, in which a broker prices model instances according to predictive quality while preventing arbitrage. Liu \textit{et al.} \cite{liu2021dealer} proposed a broker-centric model marketplace that incorporates differential privacy and Shapley-value-based compensation for data contributors.
With the rise of federated learning, several studies further explored decentralized model marketplaces.
Sun \textit{et al.} \cite{sun2022profit} designed a DPFL-based marketplace that explicitly accounts for data owners' privacy costs, while Pan \textit{et al.} \cite{pan2023privacy} addressed trainer selection under non-IID data and heterogeneous task requirements.

Despite these advances, existing model marketplaces largely treat models as indivisible commodities and focus on pricing, privacy, and contribution from datasets.
They do not consider scenarios in which multiple heterogeneous models are dynamically coordinated to jointly provide a composite service.

\noindent
\textbf{Mixture of Experts.}
The Mixture-of-Experts (MoE) architecture employs multiple specialized models together with a gating network to dynamically select and combine expert predictions, enabling scalable and efficient learning across diverse domains \cite{yuksel2012twenty,mu2025comprehensive}.
MoE has been successfully applied in continual learning \cite{hihn2021mixture}, meta-learning \cite{dai2021generalizable}, multi-task learning \cite{reuss2024efficient}, reinforcement learning \cite{wei2025deep}, federated learning \cite{su2023one}, and, more recently, large language models \cite{song2024promoe,yang2025faster,sun2026multi}.

However, in conventional MoE systems, experts and gating networks are co-trained components within a single learning pipeline. To the best of our knowledge, the MoE paradigm has not been proposed as a market-level coordination mechanism among heterogeneous model providers. This paper extends MoE from an internal learning architecture to a market-level framework for coordinating heterogeneous experts, where expert participation is governed not only by model utility but also by operational cost.

\section{Problem Formulation}
In this section, we formally define the proposed \textbf{MoE-based model markets} (as depicted in Fig.~\ref{fig:proposed-model}). We formalize the participants, decision variables, and economic costs of the composite service in the proposed MoE-based model market. Table~\ref{tab:notations} summarizes the main notation used throughout the formulation.

\begin{table}[t]
  \caption{Key Notations in the MoE Market Model}
  \label{tab:notations}
  \centering
  \begin{tabular}{ccp{4.2cm}}
    \toprule
    Symbol & Type & Description\\
    \midrule
    $\mathcal{M}$ & tuple & MoE Market $(\mathcal{B},\mathcal{N},\mathcal{P},\mathcal{G},\mathcal{E})$\\
    $\mathcal{B}$ & set & Set of buyers\\
    $r=(\tau,q,\pi)$ & tuple & Buyer service request\\
    $\mathcal{N}=\{N\}$ & set & Broker set\\
    $\mathcal{P}$ & set & Set of data providers\\
    $\mathcal{E}$ & set & Set of experts\\
    $f_i(x;\theta_i)$ & function & Predictive model of expert $E_i$\\
    $g(x,r;\phi)$ & function & Gating network routing inputs to experts\\
    $w_i(x,r)$ & variable & Gating weight of expert $E_i$\\
    $\mathcal{D}_r$ & distribution & Data distribution induced by request $r$\\
    $c_i(x,r)$ & function & Participation cost of expert $E_i$\\
    $\ell_\tau(\cdot,\cdot)$ & function & Task-specific prediction loss\\
    $U(r)$ & variable & Utility of delivered composite model\\
    $\mathrm{Welfare}(r)$ & variable & Market welfare objective\\
    $\pi_r$ & variable & Service price paid by the buyer\\
    $\kappa_r$ & variable & Broker coordination fee\\
    $R_r$ & variable & Residual revenue pool for experts\\
    $a_i(r)$ & variable & Revenue allocated to expert $E_i$\\
    \bottomrule
  \end{tabular}
\end{table}

\subsection{Market Participants}
\label{market_participants}

\begin{definition}[MoE Market]
The proposed MoE Market is defined as a tuple
\begin{equation}
\mathcal{M} = (\mathcal{B}, \mathcal{N}, \mathcal{P}, \mathcal{G}, \mathcal{E}),
\end{equation}
where:
\begin{itemize}
\item $\mathcal{B}$ denotes the set of buyers;
\item $\mathcal{N}=\{N\}$ denotes the set of brokers;
\item $\mathcal{P}$ denotes the set of data providers;
\item $\mathcal{G}$ denotes the gating network operated by the broker;
\item $\mathcal{E}$ denotes the set of heterogeneous experts.
\end{itemize}
\end{definition}
\begin{definition}[Buyers and Requests]
Let the set of buyers be
\begin{equation}
\mathcal{B}=\{B_1,B_2,\dots,B_{|\mathcal{B}|}\}.
\end{equation}
Each buyer $B_k\in\mathcal{B}$ is characterized by a model service request
\begin{equation}
r_k=(\tau_k,q_k,\pi_k),
\end{equation}
where $\tau_k$ specifies the task type, such as classification or regression. $q_k$ denotes a quality requirement (e.g., target accuracy or error tolerance), and $\pi_k$ represents the buyer’s budget constraint, namely the maximum amount the buyer is willing to pay for the requested service.
The request $r_k$ constitutes the buyer’s demand in the MoE Market. In particular, $q_k$ and $\pi_k$ jointly characterize the feasibility of the request. A service is acceptable only if the delivered composite model satisfies the required quality level and the negotiated service price does not exceed the buyer's budget. Otherwise, the request must be rejected or renegotiated.
\end{definition}

\begin{figure}[t]
    \centering
    \includegraphics[width=\linewidth]{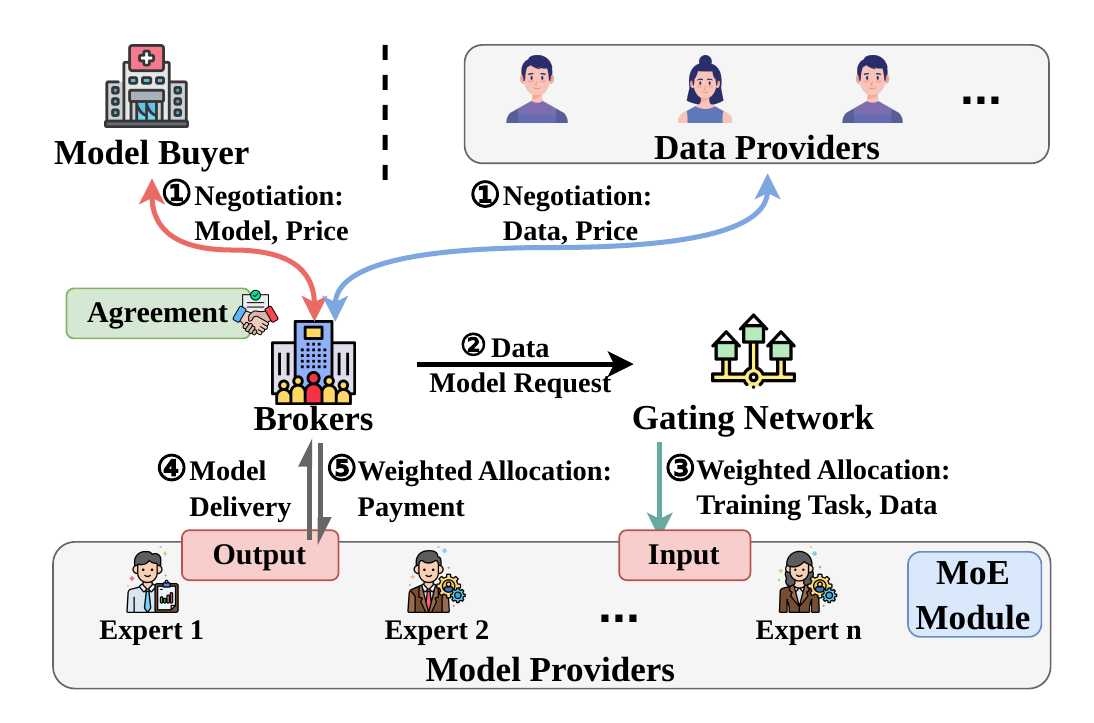}
    \caption{MoE Model Market}
    \label{fig:proposed-model}
\end{figure}

\begin{definition}[Broker]
We denote the broker by $\mathcal{N}=\{N\}$.

The broker is the market operator that receives buyer requests, initiates the transaction outcome, acquires task-relevant datasets, coordinates experts participation through the gating network, and delivers the resulting composite model service. In this paper, we model the broker as a single decision-making agent without loss of generality.

Formally, given a buyer request $r_k=(\tau_k,q_k,\pi_k)$, the broker produces a transaction outcome
\begin{equation}
N:\ r_k \mapsto (\hat{\pi}_k,\hat{f}_k,\mathbf{a}_k),
\end{equation}
where $\hat{\pi}_k$ denotes the negotiated service price, $\hat{f}_k$ denotes the delivered composite model/service, and
\begin{equation}
\mathbf{a}_k=(a_{k,1},\dots,a_{k,|\mathcal{E}|})
\end{equation}
is the payment allocation vector over experts, with $a_{k,i}$ denoting the payment assigned to expert $E_i$.

Here, $\pi_k$ is the buyer's budget limit specified in the request, whereas $\hat{\pi}_k$ is the realized transaction price after negotiation. In a feasible transaction, the negotiated price should satisfy $\hat{\pi}_k \leq \pi_k$. 

To construct $\hat{f}_k$, the broker first acquires task-relevant datasets from a subset of data providers $\mathcal{P}_k\subseteq\mathcal{P}$:
\begin{equation}
N:\ (\tau_k,q_k)\mapsto \{D_{k,m}\}_{m\in\mathcal{P}_k},
\end{equation}
where $D_{k,m}$ denotes the dataset obtained from provider $P_m$ for request $k$. The collected data and request context are then forwarded to the gating network for expert coordination.

\end{definition}

\begin{definition}[Data Providers]
Let the set of data providers be
\begin{equation}
\mathcal{P}=\{P_1,P_2,\dots,P_{|\mathcal{P}|}\}.
\end{equation}
Each $P_m\in\mathcal{P}$ denotes a data-supplying entity, which may represent an individual owner, an organization, or another abstract source of task-relevant data.

For a buyer request $r_k$, each data provider $P_m\in\mathcal{P}$ is characterized by a pricing function that maps the broker's data request to a price quote:
\begin{equation}
P_m:\ \mathrm{Req}(r_k)\mapsto \rho_{k,m},
\end{equation}
where $\rho_{k,m}$ denotes the quoted price for supplying task-relevant data for request $k$. In the current formulation, if the broker procures data from provider $P_m$, the provider is compensated according to this quoted price. That is, data providers are paid through direct procurement, whereas the revenue allocation mechanism studied later in this paper focuses on the participating experts.

Given a selected subset of providers $\mathcal{P}_k\subseteq\mathcal{P}$, the broker acquires datasets
\begin{equation*}
D_{k,m}\sim P_m,\quad m\in\mathcal{P}_k,
\end{equation*}
and constructs the aggregated training dataset
\begin{equation*}
\mathcal{D}_k=\bigcup_{m\in\mathcal{P}_k} D_{k,m}.
\end{equation*}

Here, $\mathcal{D}_k$ represents the output of an upstream data discovery and integration stage. We assume that the selected datasets have been assessed for task relevance and transformed into a compatible task schema before aggregation. The specific mechanisms for data-provider discovery, relevance assessment, and schema integration are modular to our framework.

\end{definition}


\begin{definition}[Gating Network]
The gating network is maintained by the broker and serves as a routing mechanism for coordinating expert participation. We denote
\begin{equation}
\mathcal{G}=\{g(\cdot;\phi)\},
\end{equation}
where $\phi$ represents the gating parameters.

Formally, for a buyer request $r_k$ and an input sample $x\in\mathcal{D}_k$ drawn from the aggregated dataset associated with that request, the gating network implements a stochastic routing operator:
\begin{equation}
g:\ (x,r_k)\mapsto w_k(x)\in\Delta^{|\mathcal{E}|},
\end{equation}
where
\begin{equation*}
w_k(x)=(w_{k,1}(x),\dots,w_{k,|\mathcal{E}|}(x)),
\quad
\sum_i w_{k,i}(x)=1,\quad w_{k,i}(x)\ge 0.
\end{equation*}

The vector $w_k(x)$ specifies the participation weights of experts for input $x$, thereby determining which experts are activated and their relative contributions to the composite service. In dense routing, all experts may receive nonzero weights, in sparse routing, only experts with positive effective weights are activated. The routing weights $w_k(x)$ are generated dynamically for each input--request pair rather than assigned permanently during pre-training, training learns the gating parameters $\phi$.
\end{definition}


\begin{definition}[Experts]
Let the set of heterogeneous experts be
\begin{equation*}
\mathcal{E}=\{E_1,E_2,\dots,E_{|\mathcal{E}|}\}.
\end{equation*}
Each expert $E_i\in\mathcal{E}$ is associated with a predictive model $f_i(x;\theta_i)$, and a non-negative participation cost function, which will be included in section \ref{cost}. Each expert is instantiated independently with its own model parameters, architectural configuration, training process, and execution cost. An expert may be supplied as a pretrained service or independently fitted on request-specific training data. Once the expert is included in the market for a request, its parameters $\theta_i$ remain fixed while the broker trains the gating network. Expert heterogeneity therefore arises from differences in architecture, learned predictive behavior, and execution cost; complementarity is not assumed in advance but is identified through the learned routing policy.

Given a request $r_k$ and an input sample $x$, the set of activated experts is determined by the gating weights:
\begin{equation}
\mathcal{E}_k(x)=\{E_i\in\mathcal{E}\mid w_{k,i}(x)>0\}.
\end{equation}

The composite output is obtained by aggregating expert predictions. For classification, we adopt a probability-level mixture:
\begin{equation}
p_k(y\mid x)=\sum_{i=1}^{|\mathcal{E}|} w_{k,i}(x)\,p_i(y\mid x).
\end{equation}

Only activated experts participate in execution and are eligible for compensation. If $w_{k,i}(x)=0$, expert $E_i$ neither contributes to the composite service nor receives payment for request $r_k$.
\end{definition}

\subsection{MoE Market Workflow}

\begin{definition}[Market Workflow]
For each buyer request $r_k=(\tau_k,q_k,\pi_k)$ submitted by buyer $B_k$, the MoE Market operates according to the following steps:

\textbf{Step 1: Request Submission.}
The buyer submits a model service request $r_k$ to the broker, specifying the task type, quality requirement, and budget constraint.

\textbf{Step 2: Data Acquisition.}
Given $r_k$, the broker solicits task-relevant datasets from a subset of data providers $\mathcal{P}_k\subseteq\mathcal{P}$ and constructs the aggregated dataset
\begin{equation}
\mathcal{D}_k=\bigcup_{m\in\mathcal{P}_k} D_{k,m}.
\end{equation}

\textbf{Step 3: Expert Coordination.}
For each input instance $x\in\mathcal{D}_k$, where $x$ denotes one request-specific data instance such as a tabular feature vector or an image input, the broker invokes the gating network to produce an expert participation vector
\begin{equation}
w_k(x)=g(x,r_k;\phi)\in\Delta^{|\mathcal{E}|},
\end{equation}
which determines the activated experts and their relative contributions.

\textbf{Step 4: Composite Service Generation.}
The activated experts jointly produce a composite model output through weighted aggregation. For classification tasks, this is instantiated as
\begin{equation}
p_k(y\mid x)=\sum_{i=1}^{|\mathcal{E}|} w_{k,i}(x)\,p_i(y\mid x).
\end{equation}

For regression tasks, the composite prediction can be written as
\begin{equation}
\hat{y}_k(x)=\sum_{i=1}^{|\mathcal{E}|} w_{k,i}(x)\,\hat{y}_i(x).
\end{equation}

The resulting composite predictor constitutes the delivered service $\hat{f}_k$. Here, $p_k(y\mid x)$ denotes the predictive probability of the composite model for classification, and $\hat{y}_k(x)$ denotes the aggregated prediction for regression.

\textbf{Step 5: Service Delivery and Payment.}
If the broker can successfully deliver a feasible composite service satisfying the buyer's quality requirement and budget constraint, the buyer pays the negotiated service price $\hat{\pi}_k$ to the broker. Let $C(r_k)$ denote the expected expert execution cost associated with request $r_k$. The broker retains a coordination fee $\kappa_k$, and the residual surplus
\begin{equation}
R_k=\hat{\pi}_k-\kappa_k-C(r_k)
\end{equation}
forms the allocatable revenue pool. Here, feasibility requires $\hat{\pi}_k\leq \pi_k$ and a non-negative residual pool, i.e., $R_k\geq 0$. If no such feasible transaction can be achieved, the request must be rejected or renegotiated.

\textbf{Step 6: Revenue Allocation.}
The broker allocates the residual revenue pool $R_k$ among participating experts according to the revenue allocation mechanism defined in Section~\ref{revenue}. Only experts activated under the gating outcome are eligible to receive a positive allocation for the corresponding request.
\end{definition}

\subsection{Cost Model}
\label{cost}

\begin{definition}[Expert Cost]

Each expert $E_i\in\mathcal{E}$ is associated with a participation cost function:
\begin{equation}
c_i(x,r):\mathcal{X}\times\mathcal{R}\rightarrow \mathbb{R}_{\ge 0},
\end{equation}
where $c_i(x,r)$ denotes the economic or operational cost incurred when expert $E_i$ is invoked on input $x$ under request $r$.
This cost function captures heterogeneous factors such as computation, latency, resource usage, opportunity cost, or internal service charges. Task-dependent preprocessing costs can also be incorporated. Preprocessing overhead specific to expert $E_i$ can be included in $c_i(x,r)$, whereas preprocessing shared by all experts can be treated as broker-side operational overhead covered by the coordination fee $\kappa_r$.
\end{definition}

\begin{definition}[Instance-level Service Cost]
Given a request $r$ and an input $x$, the service cost incurred for processing $x$ is defined as
\begin{equation}
\mathrm{Cost}(x,r)=\sum_{i=1}^{|\mathcal{E}|} w_i(x,r)\,c_i(x,r),
\end{equation}
where $w_i(x,r)$ denotes the gating weight assigned to expert $E_i$.
The gating weight is interpreted as the fraction of a divisible inference workload, such as query volume, token usage, or computational effort, assigned to an expert. Consequently, $w_i(x,r)c_i(x,r)$ represents the variable usage-dependent cost attributable to that expert. The present formulation abstracts from fixed invocation or API start-up fees. Such fees could be incorporated through an additional term of the form $b_i\mathbb{I}[w_i(x,r)>0]$, but are left for future work.
\end{definition}

\begin{definition}[Expected Service Cost]
Let $\mathcal{D}_r$ denote the data distribution induced by request $r$.
The expected service cost under request $r$ is defined as
\begin{equation}
\label{ExpectedServiceCost}
\mathbb{E}[\mathrm{Cost}\mid r]
=
\mathbb{E}_{x\sim\mathcal{D}_r}
\left[
\sum_{i=1}^{|\mathcal{E}|} w_i(x,r)\,c_i(x,r)
\right].
\end{equation}
\end{definition}

\begin{definition}[Budget Constraint]
When a buyer request specifies a budget limit $\pi_r$, the broker is required to satisfy the following constraint:
\begin{equation*}
\mathbb{E}[\mathrm{Cost}\mid r]\le \pi_r.
\end{equation*}
\end{definition}

\newtheorem{assumption}{Assumption}
\begin{assumption}[Bounded Cost and Measurability]
For every expert $E_i\in\mathcal{E}$ and request $r\in\mathcal{R}$, the participation cost function $c_i(x,r)$ is measurable in $x$ and uniformly bounded, i.e., there exists a constant $C_{\max}>0$ such that
\begin{equation}
0 \le c_i(x,r) \le C_{\max},
\quad \forall x\in\mathcal{X},\ \forall r\in\mathcal{R},\ \forall E_i\in\mathcal{E}.
\end{equation}
Moreover, the gating weights satisfy
\begin{equation}
0 \le w_i(x,r) \le 1,
\quad \forall x\in\mathcal{X},\ \forall r\in\mathcal{R},\ \forall E_i\in\mathcal{E}.
\end{equation}
and
\begin{equation}
\sum_{i=1}^{|\mathcal{E}|} w_i(x,r)=1,
\quad \forall x\in\mathcal{X},\ \forall r\in\mathcal{R}.
\end{equation}
\end{assumption}

\subsection{Utility and Welfare Objective}
\label{utility_and_welfare_objective}

Under a given buyer request $r$, the broker trains the gating network to balance predictive performance and expert participation cost.
We formalize this trade-off via a utility-based welfare objective. Throughout this formulation, $\mathcal{D}_r$ denotes the target deployment distribution associated with buyer request $r$. We assume that the request-specific aggregated data used for training and validation are representative samples from this distribution. Consequently, the welfare objective evaluates predictive utility and service cost under the distribution on which the delivered service is expected to operate. The objective is not intended to model all monetary transfers in the market; rather, it provides a trainable criterion that combines the predictive value of the delivered service with the operational cost of expert participation.

\begin{definition}[Predictive Loss]
Let $\ell_\tau(\cdot,\cdot)$ denote the task-specific loss induced by request $r$.
Given a labeled sample $(x,y)$ and composite output $\hat{y}(x)$, the expected predictive loss is defined as
\begin{equation}
\mathcal{L}_{\mathrm{pred}}(r)=
\mathbb{E}_{(x,y)\sim\mathcal{D}_r}
\big[\ell_\tau(y,\hat{y}(x))\big].
\end{equation}
For classification tasks, $\ell_\tau$ corresponds to cross entropy loss, while for regression tasks it is instantiated as mean squared error.
\end{definition}

\begin{definition}[Utility Function]
Given a request $r$, we define the utility of the delivered composite model service as an exponential transformation of the task loss:
\begin{equation}
U(r)
=
\alpha\cdot
\mathbb{E}_{(x,y)\sim\mathcal{D}_r}
\left[
\exp\!\left(-\frac{\ell_\tau(y,\hat{y}(x)}{\ell_{\mathrm{ref}}}\right)
\right],
\end{equation}
where $\ell_\tau(\cdot,\cdot)$ denotes the task-specific loss induced by request $r$, $\hat{y}(x)$ denotes the composite model output, $\ell_{\mathrm{ref}}>0$ is a reference loss for normalization, and $\alpha>0$ controls utility scaling.
\end{definition}

\begin{definition}[Market Welfare]
Given request $r$, we define the market welfare as the utility penalized by the expected service cost Eq.\eqref{ExpectedServiceCost}:
\begin{equation}
\label{welfare_r}
\mathrm{Welfare}(r)
=
U(r)
-
\beta\cdot
\mathbb{E}_{x\sim\mathcal{D}_r}
\Big[
\sum_{i=1}^{|\mathcal{E}|} w_i(x,r)\,c_i(x,r)
\Big],
\end{equation}
where $\beta\ge0$ controls the trade-off between predictive quality and cost sensitivity.
\end{definition}

\begin{definition}[Learning Objective]
Given independently trained experts with fixed parameters $\{\theta_i\}_{i=1}^{|\mathcal{E}|}$, the broker learns the gating parameters $\phi$ by maximizing the welfare under request $r$ in Eq.~\eqref{welfare_r}:
\begin{equation}
\max_{\phi}\ \mathrm{Welfare}(r).
\end{equation}

Equivalently, the gating network is trained by minimizing the following market-aware loss:
\begin{equation}
\label{eq:L_market}
\begin{aligned}
\min_{\phi}\ \mathcal{L}_{\mathrm{market}}
=
&-
\alpha \cdot
\mathbb{E}_{(x,y)\sim\mathcal{D}_r}
\left[
\exp\!\left(
-\frac{\ell_\tau(y,\hat{y}(x))}
{\ell_{\mathrm{ref}}}
\right)
\right]\\
&+\beta\cdot
\mathbb{E}_{x\sim\mathcal{D}_r}
\left[
\sum_{i=1}^{|\mathcal{E}|}
w_i(x,r)c_i(x,r)
\right],
\end{aligned}
\end{equation}
where $\hat{y}(x)$ is obtained by combining the fixed expert predictions according to the learned gating weights, and $\ell_{\mathrm{ref}}>0$ is a reference loss used for normalization.

The first term rewards routing policies that produce high-quality composite predictions, while the second penalizes costly expert usage. Thus, minimizing $\mathcal{L}_{\mathrm{market}}$ learns a cost-aware coordination policy over independently trained experts rather than modifying the experts themselves. It is the optimization form induced directly by the market-welfare objective.
\end{definition}

\begin{definition}[Revenue Pool]
For each buyer request $r$, the buyer pays a service price $\hat{\pi}_r$ to the broker upon successful delivery of the composite model service.
The broker retains a coordination fee $\kappa_r\ge 0$ to cover platform operation and coordination overhead.

Let
\begin{equation}
C_i(r)=\mathbb{E}_{x\sim\mathcal{D}_r}\big[w_i(x,r)c_i(x,r)\big]
\end{equation}
denote the expected participation cost of expert $E_i$ under request $r$.
The total expected expert execution cost is therefore
\begin{equation}
C(r)=\sum_{i=1}^{|\mathcal{E}|} C_i(r)
=
\mathbb{E}_{x\sim\mathcal{D}_r}
\left[
\sum_{i=1}^{|\mathcal{E}|} w_i(x,r)c_i(x,r)
\right].
\end{equation}

The remaining amount
\begin{equation}
R_r=\hat{\pi}_r-\kappa_r-C(r)
\end{equation}
is defined as the residual surplus available for allocation among participating experts.
In this work, we assume $R_r\ge 0$.
\end{definition}

Under this formulation, the broker acts as a platform operator that retains a coordination fee, while the residual revenue pool $R_r$ is redistributed among experts according to their service contributions.
Market welfare captures the net system value generated by expert coordination, whereas $R_r$ provides the monetary basis for compensating participating experts. The corresponding revenue allocation mechanism is introduced in Section~\ref{revenue}.

\section{Revenue Allocation}
\label{revenue}
Upon completion of a composite model service for request $r$, the broker allocates the residual revenue pool $R_r$ among participating experts according to their contributions to the delivered service.
This section formalizes expert payoffs and presents a cost-aware allocation rule aligned with the MoE coordination mechanism. The allocation rule applies to the residual surplus after the broker coordination fee and expected expert execution costs have been accounted for.

\begin{definition}[Cost-Adjusted Revenue Allocation]
For a given request $r$, let $R_r$ denote the residual revenue pool after broker coordination fees and expected expert execution costs. We define the cost-adjusted contribution score of expert $E_i$ as
\begin{equation}
s_i(r)
=
\mathbb{E}_{x\sim\mathcal{D}_r}
\Big[
w_i(x,r)\cdot \exp\!\big(-\lambda\, c_i(x,r)\big)
\Big],
\end{equation}
where $w_i(x,r)$ captures the realized participation of expert $E_i$ under the gating outcome, and the exponential factor discounts this contribution according to the expert's execution cost. The parameter $\lambda>0$ controls the strength of this discount: larger $\lambda$ imposes a stronger penalty on costly experts, whereas smaller $\lambda$ makes the allocation closer to participation-proportional sharing. The exponential form is also convenient for analysis, as it yields clear limiting behavior and monotone cost sensitivity, which are formalized in the propositions below.

The revenue allocated to expert $E_i$ under request $r$ is then defined as
\begin{equation}
a_i(r)
=
R_r\cdot
\frac{s_i(r)}{\sum_{j=1}^{|\mathcal{E}|}s_j(r)}.
\end{equation}
Thus, the proposed rule distributes the residual revenue proportionally to cost-adjusted contribution scores rather than raw participation alone. Only experts with $s_i(r)>0$ receive positive revenue.
\end{definition}

This design turns the allocation objective into an operational payment rule: experts are rewarded for actual contribution to the delivered composite service, but their compensation is explicitly discounted when achieving that contribution requires higher execution cost. 

\begin{proposition}[Budget Balance]
For any request $r$ such that $\sum_{j=1}^{|\mathcal{E}|} s_j(r)>0$, the proposed allocation is budget-balanced:
\begin{equation}
\sum_{i=1}^{|\mathcal{E}|} a_i(r)=R_r.
\end{equation}
\end{proposition}

\begin{proof}
By definition,
\begin{equation}
\sum_{i=1}^{|\mathcal{E}|} a_i(r)
=
\sum_{i=1}^{|\mathcal{E}|}
R_r\cdot\frac{s_i(r)}{\sum_{j=1}^{|\mathcal{E}|} s_j(r)}
=
R_r\cdot
\frac{\sum_{i=1}^{|\mathcal{E}|} s_i(r)}{\sum_{j=1}^{|\mathcal{E}|} s_j(r)}
=
R_r.
\end{equation}
\end{proof}

\begin{proposition}[Participation Monotonicity]
For any fixed request $r$ and residual revenue $R_r> 0$, if two experts satisfy $s_i(r)\ge s_j(r)$, then $a_i(r)\ge a_j(r)$.
Moreover, if $s_i(r)>s_j(r)$, $\sum_k s_k(r)>0$, and $R_r> 0$, then $a_i(r)>a_j(r)$.
\end{proposition}

\begin{proof}
Since $a_i(r)=R_r\, s_i(r)/\sum_k s_k(r)$ and $R_r\, s_i(r)/\sum_k s_k(r)>0$, the mapping $s_i\mapsto p_i$ is strictly increasing when $R_r>0$.
Thus $s_i(r)\ge s_j(r)$ implies $a_i(r)\ge a_j(r)$, with strict inequality when $s_i(r)>s_j(r)$.
\end{proof}

\begin{proposition}[Cost Sensitivity]
Fix request $r$ and assume $\sum_{k=1}^{|\mathcal{E}|} s_k(r)>0$.
For any expert $E_i$, if its cost function increases pointwise, i.e., $\tilde c_i(x,r)\ge c_i(x,r)$ for all $x$, while all other quantities remain unchanged, then its score weakly decreases:
\begin{equation}
\tilde s_i(r)\le s_i(r),
\end{equation}
and consequently its payment weakly decreases:
\begin{equation}
\tilde a_i(r)\le a_i(r).
\end{equation}
\end{proposition}

\begin{proof}
Because $\lambda>0$ and the exponential function is strictly decreasing in its argument,
$\exp(-\lambda \tilde c_i(x,r))\le \exp(-\lambda c_i(x,r))$ holds for all $x$.
Multiplying by $w_i(x,r)\ge 0$ and taking expectation yields $\tilde s_i(r)\le s_i(r)$.

Let $S_{-i}(r)=\sum_{j\ne i}s_j(r)$ denote the total score of all experts other than $E_i$, which remains unchanged when only $c_i$ changes. The allocation to expert $E_i$ can be written as
\begin{equation}
a_i(r)=R_r\frac{s_i(r)}{s_i(r)+S_{-i}(r)}.
\end{equation}
For $R_r\ge 0$ and $S_{-i}(r)\ge 0$, this expression is weakly increasing in $s_i(r)$. Since increasing $c_i$ weakly decreases $s_i(r)$, it follows that $\tilde a_i(r)\le a_i(r)$.
\end{proof}

\begin{proposition}[Limiting Cases of $\lambda$]
Fix request $r$ and assume $\mathbb{E}_{x\sim\mathcal{D}_r}[w_i(x,r)]>0$ for at least one expert.
Then:
\begin{enumerate}
\item As $\lambda\to 0$, the allocation converges to participation-proportional sharing:
\begin{equation}
a_i(r)\ \longrightarrow\ R_r\cdot \mathbb{E}_{x\sim\mathcal{D}_r}[w_i(x,r)].
\end{equation}

\item As $\lambda\to\infty$, the allocation increasingly favors lower-cost experts among those with positive participation. In the special case where expert costs are request-level constants $c_i(r)$ and there exists a unique lowest-cost expert $i^\star$ with positive expected participation, the allocation concentrates on $i^\star$, i.e., $a_{i^\star}(r)\to R_r$ and $a_j(r)\to 0$ for $j\ne i^\star$.
\end{enumerate}
\end{proposition}

\begin{proof}
(1) For each $i$, $e^{-\lambda c_i(x,r)}\to 1$ pointwise as $\lambda\to 0$ and is bounded by $1$.
By dominated convergence, $s_i(r)=\mathbb{E}[w_i e^{-\lambda c_i}]\to \mathbb{E}[w_i]$.
Because $\sum_j \mathbb{E}[w_j]=\mathbb{E}[\sum_j w_j]=1$, we obtain

\begin{equation}
a_i(r)=R_r\frac{s_i(r)}{\sum_j s_j(r)}
\to R_r\,\mathbb{E}[w_i].
\end{equation}

(2) In the constant-cost case, the score can be written as
\begin{equation}
s_i(r)=\exp(-\lambda c_i(r))\,\mathbb{E}_{x\sim\mathcal{D}_r}[w_i(x,r)].
\end{equation}
If $i^\star$ is the unique lowest-cost expert with positive expected participation, then for any $j\ne i^\star$,
\begin{equation}
\frac{s_j(r)}{s_{i^\star}(r)}
=
\frac{\mathbb{E}[w_j]}{\mathbb{E}[w_{i^\star}]}
\exp\!\left(-\lambda(c_j(r)-c_{i^\star}(r))\right)
\to 0
\end{equation}
as $\lambda\to\infty$. Therefore $a_j(r)\to 0$ for $j\ne i^\star$, and budget balance implies $a_{i^\star}(r)\to R_r$.
\end{proof}

\section{Experimental Evaluation}

In this section, we conduct a systematic experimental evaluation to validate the proposed MoE Market mechanism. Our evaluation is designed to answer three questions: whether MoE Market achieves predictive performance comparable to or better than Welfare-aware Static Selection, Uniform Average, and Standard MoE baselines; whether cost-aware routing improves market welfare under latency-derived expert costs; and whether the proposed allocation rule produces economically meaningful expert payments while offering lower computational overhead than exact Shapley allocation.

\subsection{Experimental Setup}
\label{experimental_setup}

In the experimental setting, each benchmark task is treated as a buyer request $r=(\tau,q,\pi)$, and the corresponding dataset is treated as the aggregated request-specific dataset $\mathcal{D}_r$ after data acquisition. The four independently trained experts instantiate heterogeneous model services with distinct architectures and execution costs, the gating network represents the broker's coordination mechanism, and the weighted ensemble output corresponds to the composite model service delivered to the buyer. Accordingly, each benchmark dataset is treated as the output of the upstream data discovery and integration stage. Since the experiments focus on expert coordination and allocation, we do not separately simulate data-provider discovery, relevance assessment, dataset integration, or negotiation.

\subsubsection{Datasets and tasks}

We evaluate our MoE market on a diverse collection of benchmarks covering regression, binary classification and multi-class classification tasks on tabular and image domains. In total, our experimental study includes fifteen datasets.
For each tabular dataset and random seed, we divide the available instances into 80\% training, 10\% validation, and 10\% test sets. Classification datasets are split using label stratification, whereas regression datasets are split without stratification. Numerical features are median-imputed and standardized, while categorical features are imputed using their most frequent values and one-hot encoded. Regression targets are additionally standardized using statistics estimated from the training split. All preprocessing transformations are fitted exclusively on the training data and subsequently applied to the validation and test splits. For image datasets, we retain the official benchmark test set and randomly reserve 10\% of the official training set for validation. All images are resized to $32\times32$ pixels, single-channel images are converted to three channels, and pixel values are normalized to $[-1,1]$. No data augmentation is applied. The same seed-specific data splits and processed representations are used by every expert and comparison method. This controlled protocol isolates architectural and coordination differences from differences in data access or preprocessing.
More precisely, preprocessing is fixed at the dataset level and is expert- and method-agnostic, all experts and compared methods receive the same processed representation. We do not assume that preprocessing is universally optimal or that task-dependent preprocessing cannot affect expert performance.

\paragraph{Tabular datasets.}
We use eleven tabular datasets to evaluate the proposed framework under heterogeneous conditions. The regression datasets include \textbf{Abalone} \cite{abalone_1}, \textbf{Concrete Compressive Strength} \cite{concrete_compressive_strength_165}, \textbf{Forest Fires} \cite{forest_fires_162}, \textbf{Bike Sharing} \cite{bike_sharing_275}, \textbf{Diabetes} \cite{diabetes_34}, and \textbf{Wine Quality} \cite{wine_quality_186}. The classification benchmarks include \textbf{Breast Cancer Wisconsin (Diagnostic)} \cite{breast_cancer_wisconsin}, \textbf{PageBlocks} \cite{page_blocks_classification_78}, \textbf{Yeast} \cite{yeast_110}, \textbf{Adult Census Income} \cite{adult_2}, and \textbf{Ionosphere} \cite{ionosphere_52}. These datasets vary substantially in sample size, feature dimension, label structure, and task difficulty, providing a broad test bed for evaluating cost-aware expert coordination.

\paragraph{Image datasets.}
We further evaluate the proposed framework on four widely used image classification benchmarks: \textbf{MNIST} \cite{lecun2002gradient}, \textbf{F-MNIST} \cite{xiao2017/online}, \textbf{KMNIST} \cite{clanuwat2018deep}, and \textbf{SVHN} \cite{netzer2011reading}. These datasets range from relatively simple grayscale character recognition tasks to more challenging natural image classification problems, allowing us to assess the scalability of the proposed framework across different visual domains and difficulty levels.

The tabular benchmarks are used to assess regression and classification behavior under heterogeneous structured features, while the image benchmarks evaluate the effectiveness of the framework in multi-class classification tasks.

\subsubsection{Heterogeneous Expert pools}
\label{heterogeneous expert pools}

To emulate independent heterogeneous experts, we construct an expert pool $\mathcal{E}$ consisting of four separately instantiated models with different architectural configurations and measured execution costs.  Expert independence in our experiments means that each expert is trained using its own predictive objective before the gating network is trained; gradients from the gating objective are never used to update the experts. All experts receive the same seed-specific training split and processed representation, thereby controlling data availability and isolating heterogeneity arising from model architecture and execution cost. Once fitted, the complete expert bank is frozen and shared by Standard MoE and MoE Market. Consequently, differences between the two methods arise from their gating objectives rather than from different expert models.






For tabular classification and regression tasks, the expert pool is instantiated as follows:

\begin{itemize}
    \item \textbf{Expert $E_1$:} a two-hidden-layer MLP with hidden dimensions 128 and 64 and ReLU activations;

    \item \textbf{Expert $E_2$:} a two-layer Transformer encoder with hidden dimension 64 and four attention heads, followed by mean pooling and a task-specific output head;

    \item \textbf{Expert $E_3$:} an XGBoost classifier or regressor \cite{chen2016xgboost};

    \item \textbf{Expert $E_4$:} a CatBoost classifier or regressor \cite{prokhorenkova2018catboost}.
\end{itemize}

For image classification tasks, the expert pool is instantiated as follows:

\begin{itemize}
    \item \textbf{Expert $E_1$:} an image MLP with two hidden layers of dimension 512, applied to the flattened image representation;

    \item \textbf{Expert $E_2$:} a three-layer patch Transformer with patch size $4\times4$, hidden dimension 128, and four attention heads;

    \item \textbf{Expert $E_3$:} an XGBoost classifier applied to flattened image features;

    \item \textbf{Expert $E_4$:} a CatBoost classifier applied to flattened image features.
\end{itemize}

The image experts operate on the same $3\times32\times32$ input representation, and the output layer of each expert is adapted to the number of classes. The combination of neural, attention-based, and gradient-boosting architectures introduces substantially different inductive biases without assuming in advance that any expert is universally superior or that a particular form of complementarity must occur. The gating network is responsible for learning their input-dependent suitability.

\subsubsection{Gating Network Architecture}

For tabular classification and regression tasks, the gating network is a one-hidden-layer MLP consisting of a linear layer from the input dimension to 64 hidden units, a ReLU activation, and a four-dimensional output layer. For image tasks, the gate uses two convolutional layers with 16 and 32 output channels, respectively, followed by ReLU activations, pooling, and a linear four-dimensional output layer. In both cases, softmax converts the output logits into expert-routing weights. The independently fitted experts remain frozen while the gating network is trained. Standard MoE and MoE Market use the same expert bank, gating architecture, seed-specific data split, and initialization protocol on each task.

\subsubsection{Experts Costs}
\label{experts costs}

In realistic model markets, execution cost may reflect computation time, latency, resource consumption, deployment overhead, or API usage. Rather than assigning synthetic costs manually, we derive an empirical unit-cost profile from the measured inference latency of the independently fitted experts. This produces costs that depend on both the expert architecture and the request-specific input representation.

For each dataset, we measure each expert on a training batch containing up to 128 instances. We first perform 10 warm-up passes and then record 50 timed forward passes, synchronizing CUDA execution before and after each measurement when applicable. Let $\widetilde{c}_i(r)$ denote the median measured latency in milliseconds per instance for expert $E_i$. We normalize these measurements by their mean across the expert pool:

\begin{equation*}
{
c_i(r)
=
\frac{\widetilde{c}_i(r)}
{\frac{1}{|\mathcal{E}|}\sum_{j=1}^{|\mathcal{E}|}\widetilde{c}_j(r)}.
}
\end{equation*}

The resulting cost profile has mean one and preserves the measured relative execution cost of the four experts. A single profile is measured for each dataset using the expert bank associated with the first configured random seed and is then held fixed across all methods and seed repetitions for that dataset. This provides every method with the same market environment and prevents repeated timing noise from confounding comparisons between routing policies.

\subsubsection{Implementation and hardware configuration}

All experiments were implemented in Python 3.8.20 using PyTorch 2.4.1, scikit-learn 1.3.2, XGBoost 2.1.4, and CatBoost 1.2.8. The revised experiments were conducted under CUDA 12.1 on a Windows system equipped with an AMD Ryzen 9 7845HX CPU and an NVIDIA RTX 4070 GPU. Unless otherwise stated, all methods use the same software and hardware environment, seed-specific training, validation, and test splits, independently fitted expert bank, and latency-derived cost profile.

\subsection{Baselines}
\label{baselines}

We compare the proposed \textbf{MoE Market} against three representative baselines. All methods use the same independently trained and frozen expert bank, dataset split, measured expert-cost profile, and, where applicable, gating-network architecture on each dataset.

\textbf{(1) Welfare-aware Static Selection (WSS).}
This baseline represents a conventional model market in which one model is selected for the entire buyer request. For each expert $E_i$, the broker estimates its validation welfare as

\begin{equation}
\begin{aligned}
V_i^{\mathrm{val}}(r)
={}&
\alpha\,
\mathbb{E}_{(x,y)\sim\mathcal{D}_r^{\mathrm{val}}}
\left[
\exp\!\left(
-\frac{\ell_\tau(y,\hat{y}_i(x))}
{\ell_{\mathrm{ref}}}
\right)
\right]
-
\beta_{\mathrm{deploy}}c_i(r),
\end{aligned}
\end{equation}

where $\mathcal{D}_r^{\mathrm{val}}$ is the validation split and $c_i(r)$ is the measured execution cost of expert $E_i$. The selected expert is

\begin{equation}
i_{\mathrm{WSS}}^\ast
=
\arg\max_{i\in\{1,\ldots,|\mathcal{E}|\}}
V_i^{\mathrm{val}}(r),
\qquad
\hat{y}(x)=\hat{y}_{i_{\mathrm{WSS}}^\ast}(x).
\end{equation}

The selection uses only validation performance and measured cost; no test-set outcomes are used.

\textbf{(2) Uniform Average.}
This baseline uniformly aggregates the outputs of all frozen experts. Let $\hat{o}_i(x)$ denote the predictive probability vector for classification or the point prediction for regression. Its composite output is

\begin{equation}
\hat{o}_{\mathrm{UA}}(x)
=
\frac{1}{|\mathcal{E}|}
\sum_{i=1}^{|\mathcal{E}|}\hat{o}_i(x).
\end{equation}

Uniform Average represents non-adaptive collaboration: it ignores both input-dependent expert capabilities and execution costs.

\textbf{(3) Standard MoE.}
Standard MoE learns input-dependent mixture weights over the same frozen expert bank but does not include expert costs in its training objective. For each input $x$,

\begin{equation}
\begin{aligned}
w(x)&=\operatorname{softmax}(g(x;\phi))
      \in\Delta^{|\mathcal{E}|},\\
\hat{o}_{\mathrm{SM}}(x)
&=
\sum_{i=1}^{|\mathcal{E}|}
w_i(x)\hat{o}_i(x).
\end{aligned}
\end{equation}

The gating parameters are optimized solely for predictive performance:

\begin{equation}
\min_{\phi}\ 
\mathbb{E}_{(x,y)\sim\mathcal{D}_r}
\left[
\ell_\tau\bigl(y,\hat{o}_{\mathrm{SM}}(x)\bigr)
\right].
\end{equation}

\textbf{(4) MoE Market (Ours).}
MoE Market uses the same frozen expert bank and gating-network architecture as Standard MoE, but trains the gate using the market-aware objective in Eq.~\eqref{eq:L_market}.

\subsection{Model Training}
\label{model_training}
For each dataset and random seed, the four experts are first fitted independently on the same training split. The MLP and Transformer experts are trained for 10 epochs using Adam with a learning rate of $1\times10^{-3}$, while XGBoost and CatBoost use 200 boosting iterations, a maximum depth of 6, and a learning rate of $0.05$. The neural-expert checkpoint with the lowest validation loss is retained. All experts are then frozen, ensuring that the compared coordination methods use an identical expert bank.

Standard MoE and MoE Market use gates with the same architecture and initialization. Each gate is trained separately for 10 epochs with Adam, a learning rate of $1\times10^{-3}$, and a batch size of 128. Standard MoE minimizes predictive loss, whereas MoE Market minimizes Eq.~\eqref{eq:L_market}. The measured expert-cost profile is fixed before gate training and shared across all methods. For the main experiments, we set $\beta=0.3$ and $\ell_{\mathrm{ref}}=1$, with $\alpha=5$ for tabular tasks and $\alpha=1$ for image tasks. The checkpoint minimizing the corresponding validation objective is used for test evaluation.

\begin{table*}[!tp]
  \caption{Experimental results on tabular classification tasks (mean $\pm$ standard deviation over five random seeds).}
  \label{tab:tabular_classification}
  \centering
  \setlength{\tabcolsep}{8pt}
  \begin{tabular}{llccccc}
    \toprule
    \textbf{Dataset} &
    \textbf{Method} &
    \textbf{AUC $\uparrow$} &
    \textbf{Macro-F1 $\uparrow$} &
    \textbf{Accuracy $\uparrow$} &
    \textbf{Cost $\downarrow$} &
    \textbf{Welfare $\uparrow$} \\
    \midrule

    \multirow{4}{*}{\textbf{Breast Cancer}}
      & WSS
      & $\mathbf{0.9918 \pm 0.0085}$
      & $0.9387 \pm 0.0291$
      & $0.9439 \pm 0.0260$
      & $\mathbf{0.1518 \pm 0.0000}$
      & $\underline{4.1868 \pm 0.0559}$ \\
      & UA
      & $0.9886 \pm 0.0132$
      & $0.9386 \pm 0.0320$
      & $0.9439 \pm 0.0288$
      & $1.0000 \pm 0.0000$
      & $4.0355 \pm 0.0970$ \\
      & SM
      & $\underline{0.9889 \pm 0.0109}$
      & $\mathbf{0.9620 \pm 0.0135}$
      & $\mathbf{0.9649 \pm 0.0124}$
      & $1.6120 \pm 0.1966$
      & $4.0374 \pm 0.0950$ \\
      & \textbf{MM}
      & $0.9881 \pm 0.0107$
      & $\underline{0.9539 \pm 0.0228}$
      & $\underline{0.9579 \pm 0.0200}$
      & $\underline{0.9447 \pm 0.1272}$
      & $\mathbf{4.3163 \pm 0.0829}$ \\
    \addlinespace[1pt]
    \cmidrule(lr){1-7}

    \multirow{4}{*}{\textbf{PageBlocks}}
      & WSS
      & $0.8023 \pm 0.2016$
      & $0.2609 \pm 0.0982$
      & $0.9164 \pm 0.0256$
      & $\mathbf{0.2932 \pm 0.0000}$
      & $4.4116 \pm 0.0360$ \\
      & UA
      & $\underline{0.9859 \pm 0.0070}$
      & $0.2982 \pm 0.1059$
      & $0.9084 \pm 0.0093$
      & $1.0000 \pm 0.0000$
      & $4.2554 \pm 0.0328$ \\
      & SM
      & $\mathbf{0.9869 \pm 0.0078}$
      & $\mathbf{0.8053 \pm 0.0553}$
      & $\mathbf{0.9668 \pm 0.0054}$
      & $0.9977 \pm 0.0392$
      & $\underline{4.4807 \pm 0.1188}$ \\
      & \textbf{MM}
      & $0.9082 \pm 0.1452$
      & $\underline{0.6151 \pm 0.2109}$
      & $\underline{0.9529 \pm 0.0156}$
      & $\underline{0.3180 \pm 0.0045}$
      & $\mathbf{4.6084 \pm 0.0541}$ \\
    \addlinespace[1pt]
    \cmidrule(lr){1-7}

    \multirow{4}{*}{\textbf{Yeast}}
      & WSS
      & $0.7719 \pm 0.0723$
      & $0.3446 \pm 0.1132$
      & $0.5168 \pm 0.0507$
      & $\mathbf{0.3479 \pm 0.0000}$
      & $\underline{2.1707 \pm 0.1552}$ \\
      & UA
      & $\underline{0.8638 \pm 0.0318}$
      & $0.2084 \pm 0.0396$
      & $0.4752 \pm 0.0599$
      & $1.0000 \pm 0.0000$
      & $1.4049 \pm 0.0636$ \\
      & SM
      & $\mathbf{0.8649 \pm 0.0388}$
      & $\mathbf{0.4268 \pm 0.0982}$
      & $\mathbf{0.5624 \pm 0.0198}$
      & $0.8770 \pm 0.0466$
      & $1.8221 \pm 0.0716$ \\
      & \textbf{MM}
      & $0.6469 \pm 0.0732$
      & $\underline{0.3744 \pm 0.0617}$
      & $\underline{0.5262 \pm 0.0147}$
      & $\underline{0.4898 \pm 0.0240}$
      & $\mathbf{2.3148 \pm 0.0738}$ \\
    \addlinespace[1pt]
    \cmidrule(lr){1-7}

    \multirow{4}{*}{\textbf{Adult Census}}
      & WSS
      & $0.9135 \pm 0.0015$
      & $0.7944 \pm 0.0040$
      & $0.8587 \pm 0.0030$
      & $\mathbf{0.0508 \pm 0.0000}$
      & $\underline{4.0048 \pm 0.0213}$ \\
      & UA
      & $0.9209 \pm 0.0029$
      & $0.7976 \pm 0.0034$
      & $0.8682 \pm 0.0022$
      & $1.0000 \pm 0.0000$
      & $3.6362 \pm 0.0204$ \\
      & SM
      & $\mathbf{0.9293 \pm 0.0022}$
      & $\underline{0.8140 \pm 0.0046}$
      & $\underline{0.8739 \pm 0.0024}$
      & $0.4127 \pm 0.0261$
      & $3.9786 \pm 0.0138$ \\
      & \textbf{MM}
      & $\underline{0.9287 \pm 0.0018}$
      & $\mathbf{0.8149 \pm 0.0040}$
      & $\mathbf{0.8744 \pm 0.0030}$
      & $\underline{0.1086 \pm 0.0082}$
      & $\mathbf{4.0813 \pm 0.0215}$ \\
    \addlinespace[1pt]
    \cmidrule(lr){1-7}

    \multirow{4}{*}{\textbf{Ionosphere}}
      & WSS
      & $0.9244 \pm 0.0639$
      & $0.8433 \pm 0.1045$
      & $0.8722 \pm 0.0800$
      & $\mathbf{0.2589 \pm 0.0000}$
      & $\underline{3.5220 \pm 0.2109}$ \\
      & UA
      & $0.9298 \pm 0.0637$
      & $0.8069 \pm 0.0839$
      & $0.8444 \pm 0.0639$
      & $1.0000 \pm 0.0000$
      & $3.1761 \pm 0.1331$ \\
      & SM
      & $\underline{0.9358 \pm 0.0582}$
      & $\mathbf{0.8896 \pm 0.0795}$
      & $\mathbf{0.9056 \pm 0.0639}$
      & $1.1372 \pm 0.0655$
      & $3.3655 \pm 0.2115$ \\
      & \textbf{MM}
      & $\mathbf{0.9445 \pm 0.0356}$
      & $\underline{0.8819 \pm 0.0551}$
      & $\underline{0.8944 \pm 0.0456}$
      & $\underline{0.9131 \pm 0.0258}$
      & $\mathbf{3.6524 \pm 0.1851}$ \\
    \bottomrule
  \end{tabular}

  \parbox{0.9\textwidth}{
    \footnotesize
    WSS: Welfare-aware Static Selection;\,
    UA: Uniform Average;\,
    SM: Standard MoE;\,
    MM: MoE Market.
    Bold and underlined values denote the best and second-best results, respectively.
  }
\end{table*}
\begin{table*}[!tp]
  \caption{Experimental results on tabular regression tasks (mean $\pm$ standard deviation over five random seeds).}
  \label{tab:tabular_regression}
  \centering
  \setlength{\tabcolsep}{8pt}
    \begin{tabular}{llccccc}
    \toprule
    \textbf{Dataset} &
    \textbf{Method} &
    \textbf{MSE $\downarrow$} &
    \textbf{RMSE $\downarrow$} &
    \textbf{MAE $\downarrow$} &
    \textbf{Cost $\downarrow$} &
    \textbf{Welfare $\uparrow$} \\
    \midrule

    \multirow{4}{*}{\textbf{Abalone}}
      & WSS
      & $0.6130 \pm 0.1387$
      & $0.7791 \pm 0.0874$
      & $0.5378 \pm 0.0448$
      & $\mathbf{0.4529 \pm 0.0293}$
      & $3.6121 \pm 0.0941$ \\
      & UA
      & $0.5463 \pm 0.0728$
      & $0.7378 \pm 0.0504$
      & $0.5998 \pm 0.0663$
      & $1.0000 \pm 0.0000$
      & $3.1369 \pm 0.2293$ \\
      & SM
      & $\mathbf{0.4312 \pm 0.0399}$
      & $\mathbf{0.6561 \pm 0.0304}$
      & $\underline{0.4738 \pm 0.0188}$
      & $0.7965 \pm 0.0387$
      & $\underline{3.6344 \pm 0.0466}$ \\
      & \textbf{MM}
      & $\underline{0.4671 \pm 0.0487}$
      & $\underline{0.6828 \pm 0.0350}$
      & $\mathbf{0.4648 \pm 0.0226}$
      & $\underline{0.4979 \pm 0.0119}$
      & $\mathbf{3.7906 \pm 0.0653}$ \\
    \addlinespace[1pt]
    \cmidrule(lr){1-7}

    \multirow{4}{*}{\textbf{Concrete}}
      & WSS
      & $0.5566 \pm 0.1292$
      & $0.7419 \pm 0.0884$
      & $0.6108 \pm 0.0871$
      & $\mathbf{0.4968 \pm 0.0058}$
      & $3.2127 \pm 0.2923$ \\
      & UA
      & $0.3194 \pm 0.0850$
      & $0.5616 \pm 0.0712$
      & $0.4521 \pm 0.0690$
      & $1.0000 \pm 0.0000$
      & $3.6102 \pm 0.2303$ \\
      & SM
      & $\underline{0.2666 \pm 0.0455}$
      & $\underline{0.5148 \pm 0.0444}$
      & $\underline{0.4016 \pm 0.0368}$
      & $0.8943 \pm 0.0382$
      & $\underline{3.8131 \pm 0.1212}$ \\
      & \textbf{MM}
      & $\mathbf{0.2522 \pm 0.0582}$
      & $\mathbf{0.4995 \pm 0.0586}$
      & $\mathbf{0.3754 \pm 0.0421}$
      & $\underline{0.7064 \pm 0.0251}$
      & $\mathbf{3.9421 \pm 0.1352}$ \\
    \addlinespace[1pt]
    \cmidrule(lr){1-7}

    \multirow{4}{*}{\textbf{Forest Fires}}
      & WSS
      & $3.9527 \pm 3.7964$
      & $1.8131 \pm 0.9121$
      & $\underline{0.4015 \pm 0.1302}$
      & $\mathbf{0.2483 \pm 0.0301}$
      & $\underline{4.4573 \pm 0.0609}$ \\
      & UA
      & $\mathbf{3.9237 \pm 3.7619}$
      & $\underline{1.8045 \pm 0.9134}$
      & $0.4414 \pm 0.1235$
      & $1.0000 \pm 0.0000$
      & $4.3121 \pm 0.0144$ \\
      & SM
      & $\underline{3.9237 \pm 3.7648}$
      & $\mathbf{1.8042 \pm 0.9142}$
      & $0.4363 \pm 0.1227$
      & $0.8984 \pm 0.0769$
      & $4.3504 \pm 0.0286$ \\
      & \textbf{MM}
      & $3.9679 \pm 3.7788$
      & $1.8164 \pm 0.9142$
      & $\mathbf{0.3834 \pm 0.1426}$
      & $\underline{0.7192 \pm 0.0958}$
      & $\mathbf{4.5821 \pm 0.0483}$ \\
    \addlinespace[1pt]
    \cmidrule(lr){1-7}

    \multirow{4}{*}{\textbf{Bike Sharing}}
      & WSS
      & $0.5828 \pm 0.0934$
      & $0.7614 \pm 0.0621$
      & $0.5380 \pm 0.0603$
      & $\mathbf{0.1280 \pm 0.0374}$
      & $3.6817 \pm 0.2126$ \\
      & UA
      & $0.5454 \pm 0.2501$
      & $0.7225 \pm 0.1710$
      & $0.6183 \pm 0.2204$
      & $1.0000 \pm 0.0000$
      & $3.0229 \pm 0.7856$ \\
      & SM
      & $\underline{0.0530 \pm 0.0041}$
      & $\underline{0.2301 \pm 0.0089}$
      & $\underline{0.1428 \pm 0.0046}$
      & $0.4576 \pm 0.0657$
      & $\underline{4.6657 \pm 0.0285}$ \\
      & \textbf{MM}
      & $\mathbf{0.0520 \pm 0.0032}$
      & $\mathbf{0.2281 \pm 0.0070}$
      & $\mathbf{0.1363 \pm 0.0038}$
      & $\underline{0.1415 \pm 0.0078}$
      & $\mathbf{4.7727 \pm 0.0111}$ \\
    \addlinespace[1pt]
    \cmidrule(lr){1-7}

    \multirow{4}{*}{\textbf{Diabetes}}
      & WSS
      & $0.5654 \pm 0.0546$
      & $0.7512 \pm 0.0360$
      & $0.6197 \pm 0.0306$
      & $\mathbf{0.4954 \pm 0.0000}$
      & $3.2245 \pm 0.1052$ \\
      & UA
      & $0.5361 \pm 0.0623$
      & $\underline{0.7312 \pm 0.0417}$
      & $0.5910 \pm 0.0583$
      & $1.0000 \pm 0.0000$
      & $3.1687 \pm 0.1913$ \\
      & SM
      & $\mathbf{0.4893 \pm 0.0466}$
      & $\mathbf{0.6989 \pm 0.0330}$
      & $\mathbf{0.5604 \pm 0.0470}$
      & $0.9196 \pm 0.0402$
      & $\underline{3.2352 \pm 0.1041}$ \\
      & \textbf{MM}
      & $\underline{0.5360 \pm 0.0305}$
      & $0.7319 \pm 0.0206$
      & $\underline{0.5821 \pm 0.0244}$
      & $\underline{0.8700 \pm 0.0519}$
      & $\mathbf{3.2760 \pm 0.1476}$ \\
    \addlinespace[1pt]
    \cmidrule(lr){1-7}

    \multirow{4}{*}{\textbf{Wine Quality}}
      & WSS
      & $0.8400 \pm 0.0812$
      & $0.9156 \pm 0.0455$
      & $0.7101 \pm 0.0289$
      & $\mathbf{0.4311 \pm 0.0742}$
      & $3.0229 \pm 0.0851$ \\
      & UA
      & $0.6896 \pm 0.0695$
      & $0.8296 \pm 0.0416$
      & $0.6546 \pm 0.0317$
      & $1.0000 \pm 0.0000$
      & $2.9894 \pm 0.0927$ \\
      & SM
      & $\mathbf{0.5994 \pm 0.0648}$
      & $\mathbf{0.7733 \pm 0.0423}$
      & $\underline{0.6088 \pm 0.0286}$
      & $0.7661 \pm 0.0860$
      & $\underline{3.2029 \pm 0.0706}$ \\
      & \textbf{MM}
      & $\underline{0.6105 \pm 0.0636}$
      & $\underline{0.7805 \pm 0.0411}$
      & $\mathbf{0.6031 \pm 0.0293}$
      & $\underline{0.4931 \pm 0.0218}$
      & $\mathbf{3.3010 \pm 0.0673}$ \\
    \bottomrule
  \end{tabular}

  \parbox{0.9\textwidth}{
    \footnotesize
    WSS: Welfare-aware Static Selection;\,
    UA: Uniform Average;\,
    SM: Standard MoE;\,
    MM: MoE Market.
    Bold and underlined values denote the best and second-best results,
    respectively; rankings are determined from unrounded means.
  }
\end{table*}
\begin{table*}[!tp]
  \caption{Experimental results on image classification tasks (mean $\pm$ standard deviation over five random seeds).}
  \label{tab:image_classification}
  \centering
  \setlength{\tabcolsep}{8pt}
\begin{tabular}{llccccc}
    \toprule
    \textbf{Dataset} &
    \textbf{Method} &
    \textbf{AUC $\uparrow$} &
    \textbf{Macro-F1 $\uparrow$} &
    \textbf{Accuracy $\uparrow$} &
    \textbf{Cost $\downarrow$} &
    \textbf{Welfare $\uparrow$} \\
    \midrule

    \multirow{4}{*}{\textbf{MNIST}}
      & WSS
      & $0.9799 \pm 0.0260$
      & $0.8183 \pm 0.0073$
      & $0.8589 \pm 0.0045$
      & $\mathbf{0.0843 \pm 0.0000}$
      & $0.8311 \pm 0.0047$ \\
      & UA
      & $\underline{0.9924 \pm 0.0122}$
      & $0.4401 \pm 0.2955$
      & $0.4717 \pm 0.2580$
      & $1.0000 \pm 0.0000$
      & $0.0379 \pm 0.0419$ \\
      & SM
      & $\mathbf{0.9991 \pm 0.0003}$
      & $\mathbf{0.9680 \pm 0.0045}$
      & $\mathbf{0.9683 \pm 0.0044}$
      & $0.1580 \pm 0.1573$
      & $\underline{0.8632 \pm 0.0916}$ \\
      & \textbf{MM}
      & $0.9737 \pm 0.0100$
      & $\underline{0.8841 \pm 0.1219}$
      & $\underline{0.9034 \pm 0.0952}$
      & $\underline{0.0933 \pm 0.0032}$
      & $\mathbf{0.9053 \pm 0.0446}$ \\
    \addlinespace[1pt]
    \cmidrule(lr){1-7}

    \multirow{4}{*}{\textbf{Fashion-MNIST}}
      & WSS
      & $\underline{0.9872 \pm 0.0041}$
      & $\underline{0.8554 \pm 0.0339}$
      & $\underline{0.8570 \pm 0.0314}$
      & $\mathbf{0.0794 \pm 0.0102}$
      & $\underline{0.7975 \pm 0.0339}$ \\
      & UA
      & $0.9751 \pm 0.0110$
      & $0.5215 \pm 0.1979$
      & $0.5500 \pm 0.1816$
      & $1.0000 \pm 0.0000$
      & $0.0740 \pm 0.0692$ \\
      & SM
      & $0.9660 \pm 0.0092$
      & $0.8049 \pm 0.0513$
      & $0.8239 \pm 0.0343$
      & $0.0910 \pm 0.0268$
      & $0.7956 \pm 0.0302$ \\
      & \textbf{MM}
      & $\mathbf{0.9896 \pm 0.0009}$
      & $\mathbf{0.8736 \pm 0.0063}$
      & $\mathbf{0.8740 \pm 0.0062}$
      & $\underline{0.0812 \pm 0.0219}$
      & $\mathbf{0.8078 \pm 0.0027}$ \\
    \addlinespace[1pt]
    \cmidrule(lr){1-7}

    \multirow{4}{*}{\textbf{KMNIST}}
      & WSS
      & $0.9560 \pm 0.0254$
      & $\underline{0.7798 \pm 0.0683}$
      & $\underline{0.7856 \pm 0.0629}$
      & $\mathbf{0.0788 \pm 0.0000}$
      & $\underline{0.7364 \pm 0.0558}$ \\
      & UA
      & $\underline{0.9710 \pm 0.0057}$
      & $0.5817 \pm 0.2131$
      & $0.5808 \pm 0.1926$
      & $1.0000 \pm 0.0000$
      & $0.0516 \pm 0.0411$ \\
      & SM
      & $0.9131 \pm 0.0180$
      & $0.7181 \pm 0.0413$
      & $0.7420 \pm 0.0284$
      & $0.0921 \pm 0.0056$
      & $0.7150 \pm 0.0282$ \\
      & \textbf{MM}
      & $\mathbf{0.9860 \pm 0.0016}$
      & $\mathbf{0.8646 \pm 0.0066}$
      & $\mathbf{0.8642 \pm 0.0068}$
      & $\underline{0.0837 \pm 0.0109}$
      & $\mathbf{0.8014 \pm 0.0103}$ \\
    \addlinespace[1pt]
    \cmidrule(lr){1-7}

    \multirow{4}{*}{\textbf{SVHN}}
      & WSS
      & $0.6905 \pm 0.0253$
      & $0.3817 \pm 0.0638$
      & $0.4994 \pm 0.0347$
      & $\mathbf{0.0883 \pm 0.0053}$
      & $0.4712 \pm 0.0347$ \\
      & UA
      & $0.9217 \pm 0.0091$
      & $0.6054 \pm 0.0506$
      & $0.6206 \pm 0.0589$
      & $1.0000 \pm 0.0000$
      & $0.0473 \pm 0.0290$ \\
      & SM
      & $\mathbf{0.9638 \pm 0.0013}$
      & $\mathbf{0.7721 \pm 0.0052}$
      & $\mathbf{0.7908 \pm 0.0056}$
      & $0.1266 \pm 0.0435$
      & $\underline{0.6841 \pm 0.0169}$ \\
      & \textbf{MM}
      & $\underline{0.9537 \pm 0.0158}$
      & $\underline{0.7432 \pm 0.0511}$
      & $\underline{0.7649 \pm 0.0462}$
      & $\underline{0.0928 \pm 0.0015}$
      & $\mathbf{0.6914 \pm 0.0423}$ \\
    \bottomrule
  \end{tabular}

  \parbox{0.9\textwidth}{
    \footnotesize
    WSS: Welfare-aware Static Selection;\,
    UA: Uniform Average;\,
    SM: Standard MoE;\,
    MM: MoE Market.
    Bold and underlined values denote the best and second-best results,
    respectively; rankings are determined from unrounded means.
  }
\end{table*}
\begin{figure*}[htbp]
    \centering
    \begin{subfigure}{0.245\linewidth}\includegraphics[width=\linewidth]{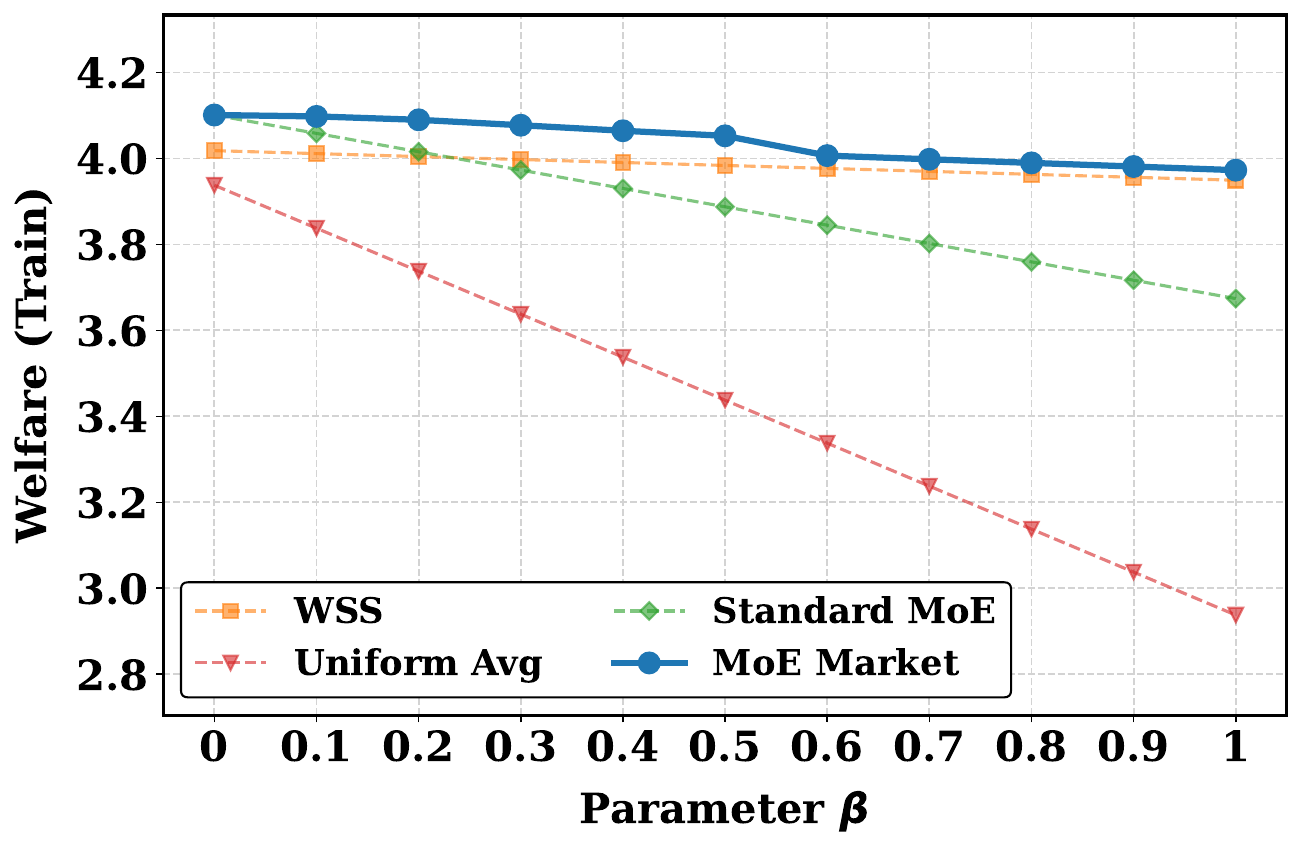}
    \end{subfigure}\hfill
    \begin{subfigure}{0.245\linewidth}\includegraphics[width=\linewidth]{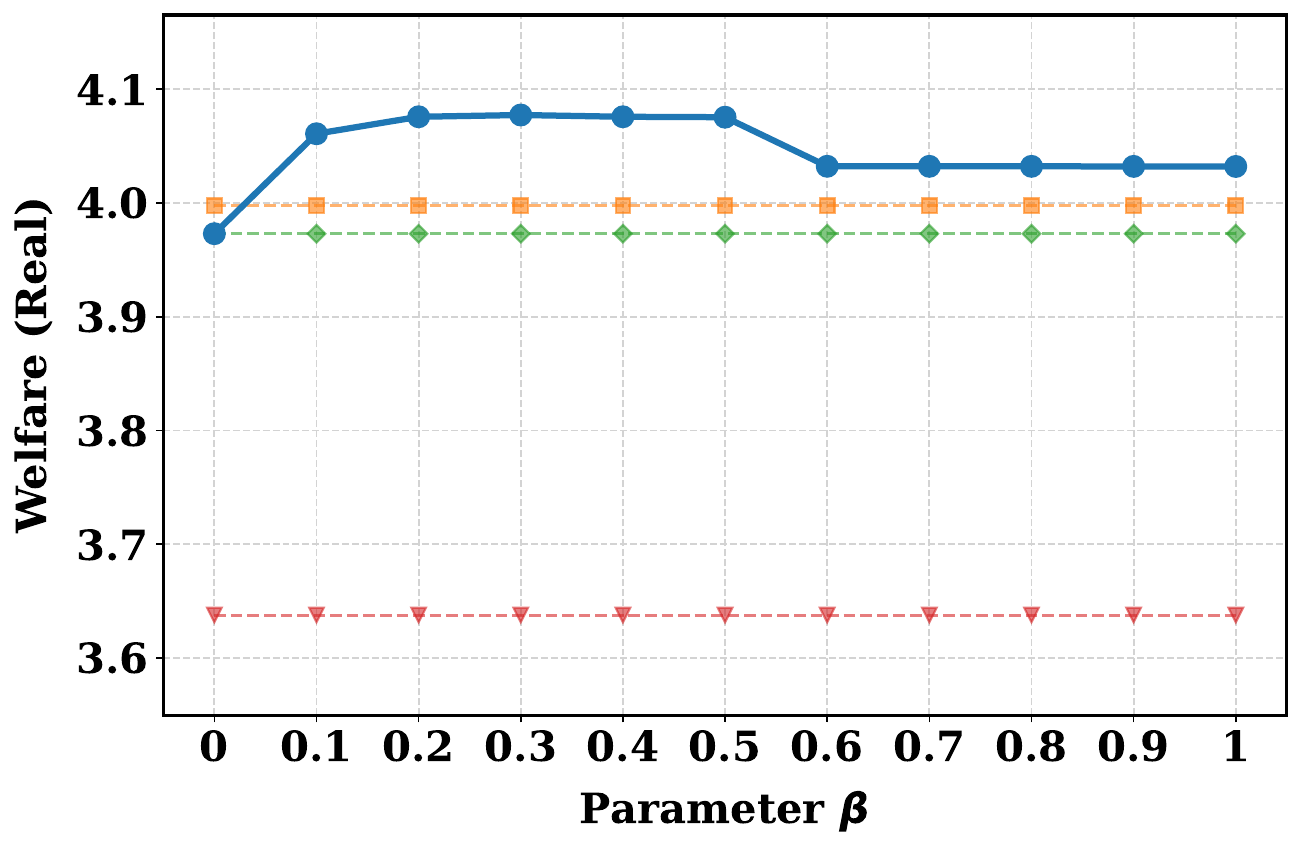}
    \end{subfigure}\hfill
    \begin{subfigure}{0.245\linewidth}\includegraphics[width=\linewidth]{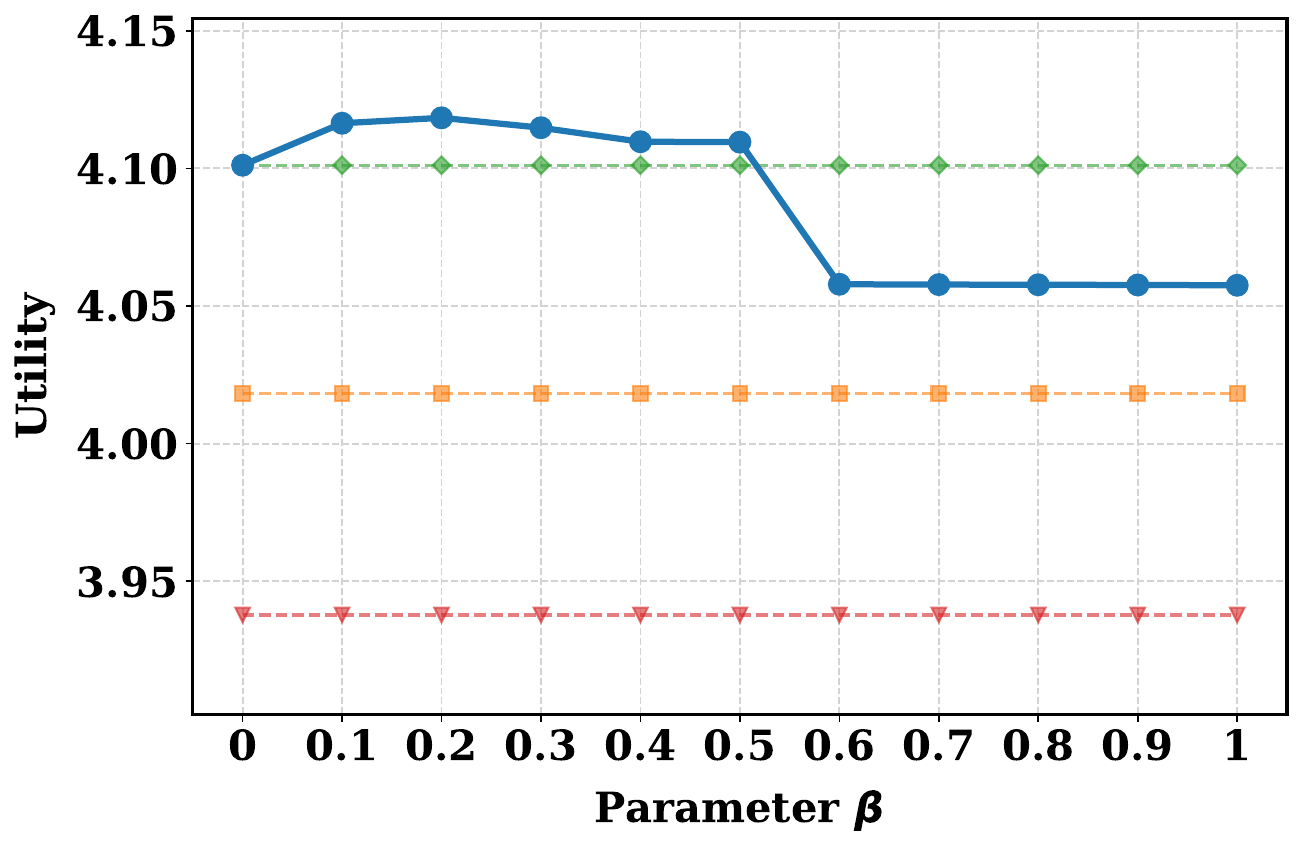}
    \end{subfigure}\hfill
    \begin{subfigure}{0.245\linewidth}\includegraphics[width=\linewidth]{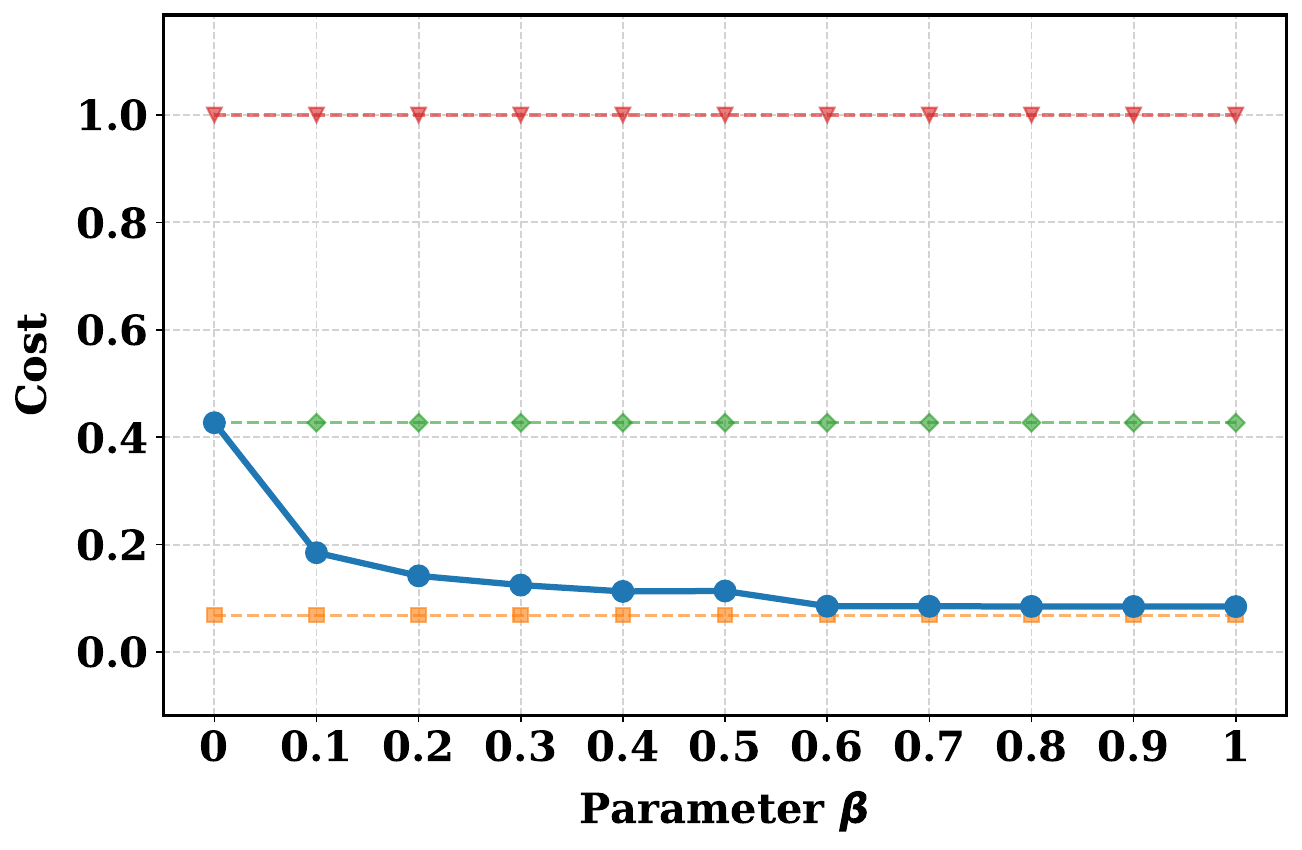}
    \end{subfigure}
    \\[2pt]
    \begin{subfigure}{0.245\linewidth}\includegraphics[width=\linewidth]{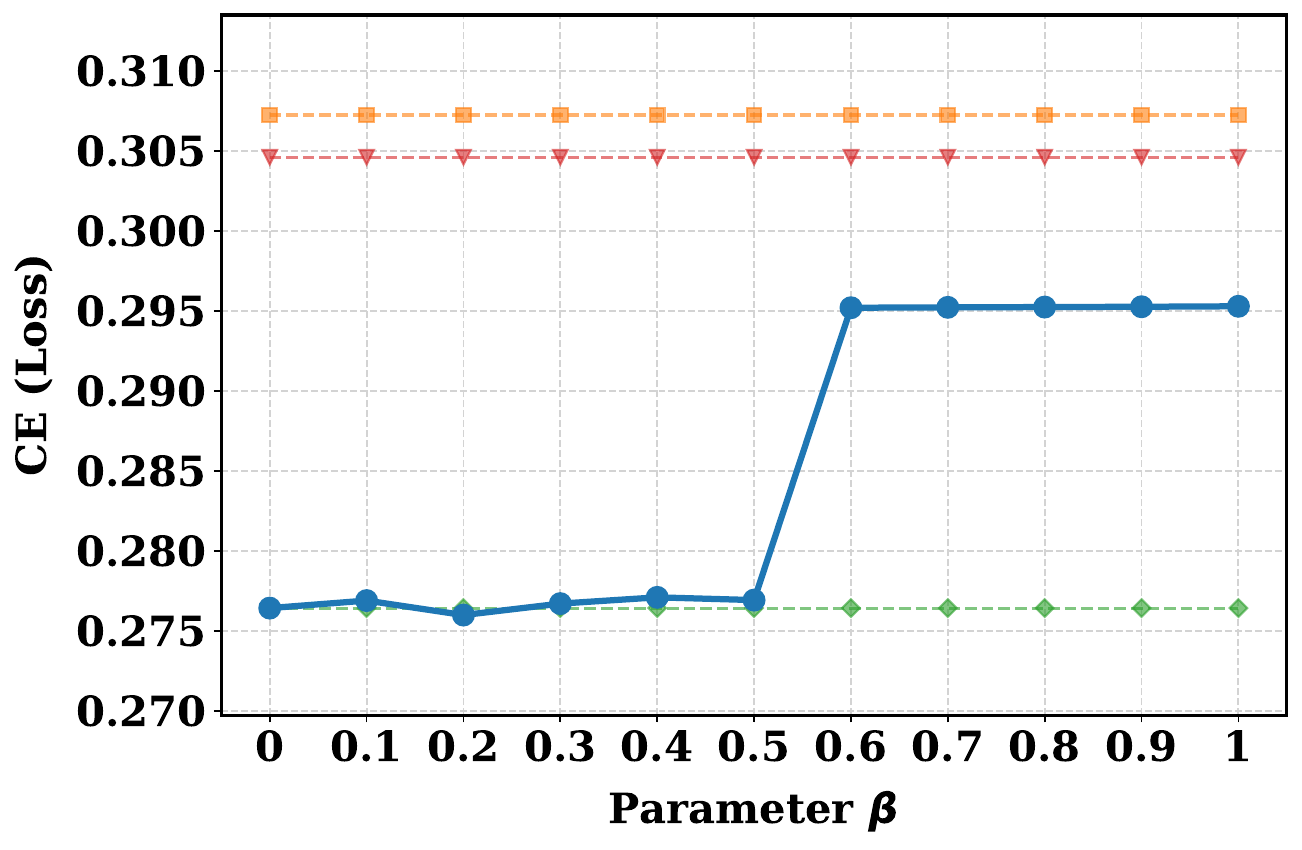}
    \end{subfigure}\hfill
    \begin{subfigure}{0.245\linewidth}\includegraphics[width=\linewidth]{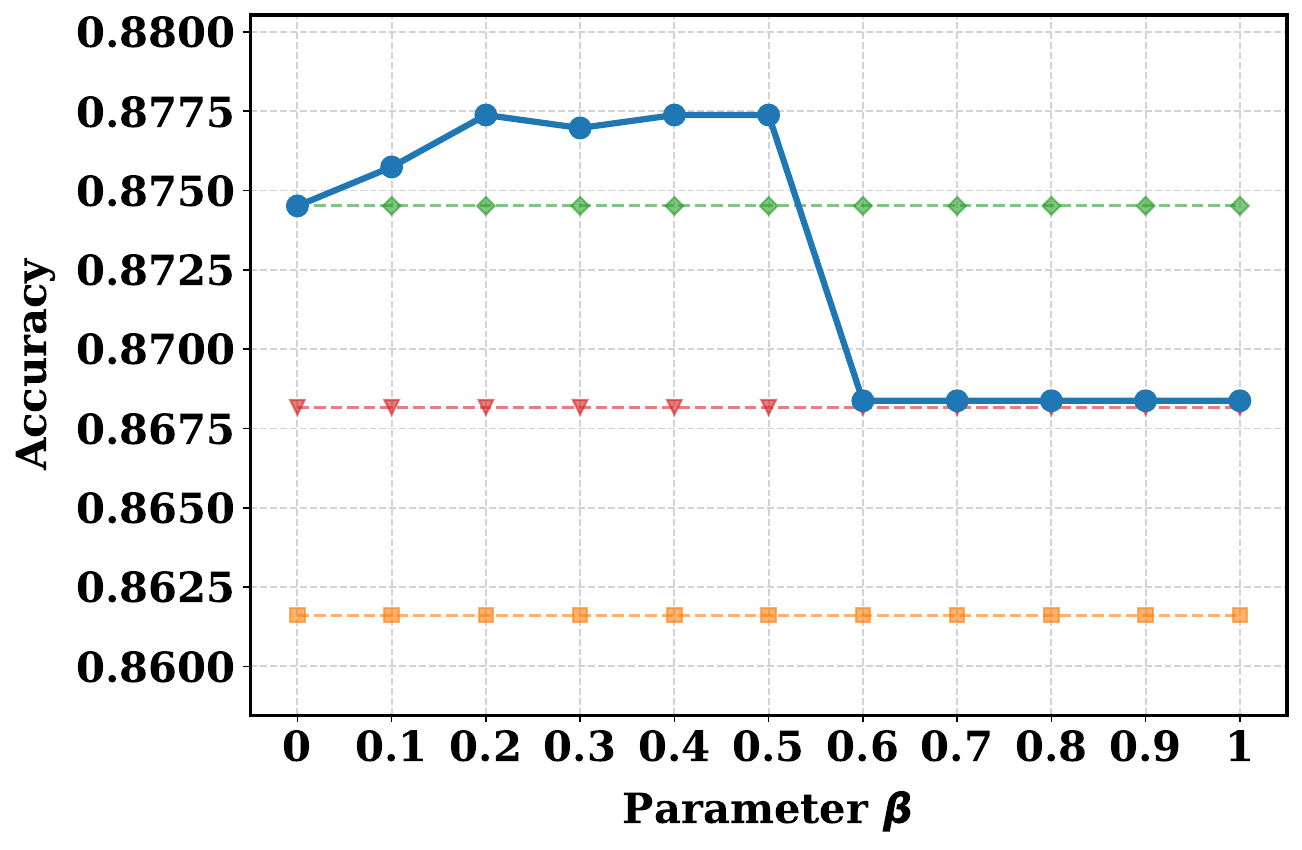}
    \end{subfigure}\hfill
    \begin{subfigure}{0.245\linewidth}\includegraphics[width=\linewidth]{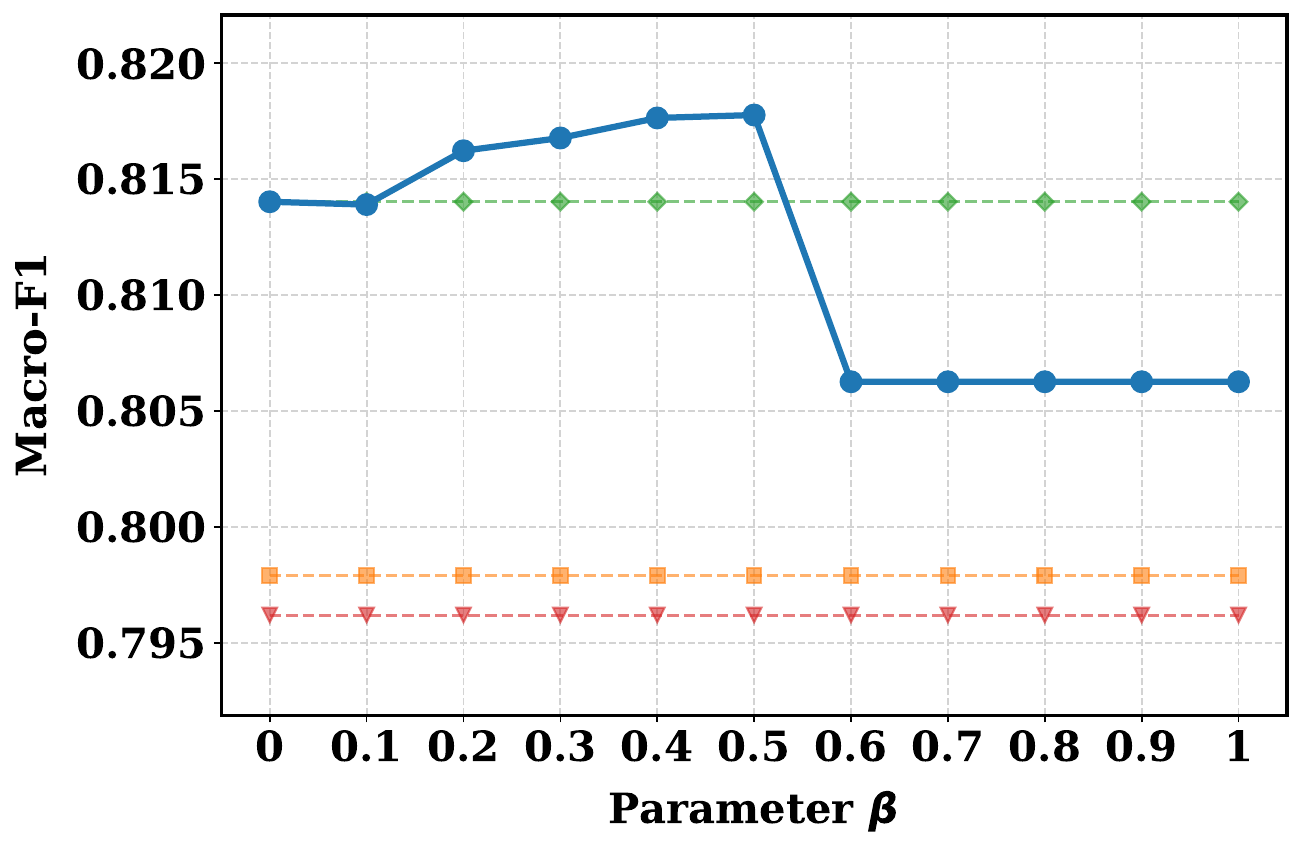}
    \end{subfigure}\hfill
    \begin{subfigure}{0.245\linewidth}\includegraphics[width=\linewidth]{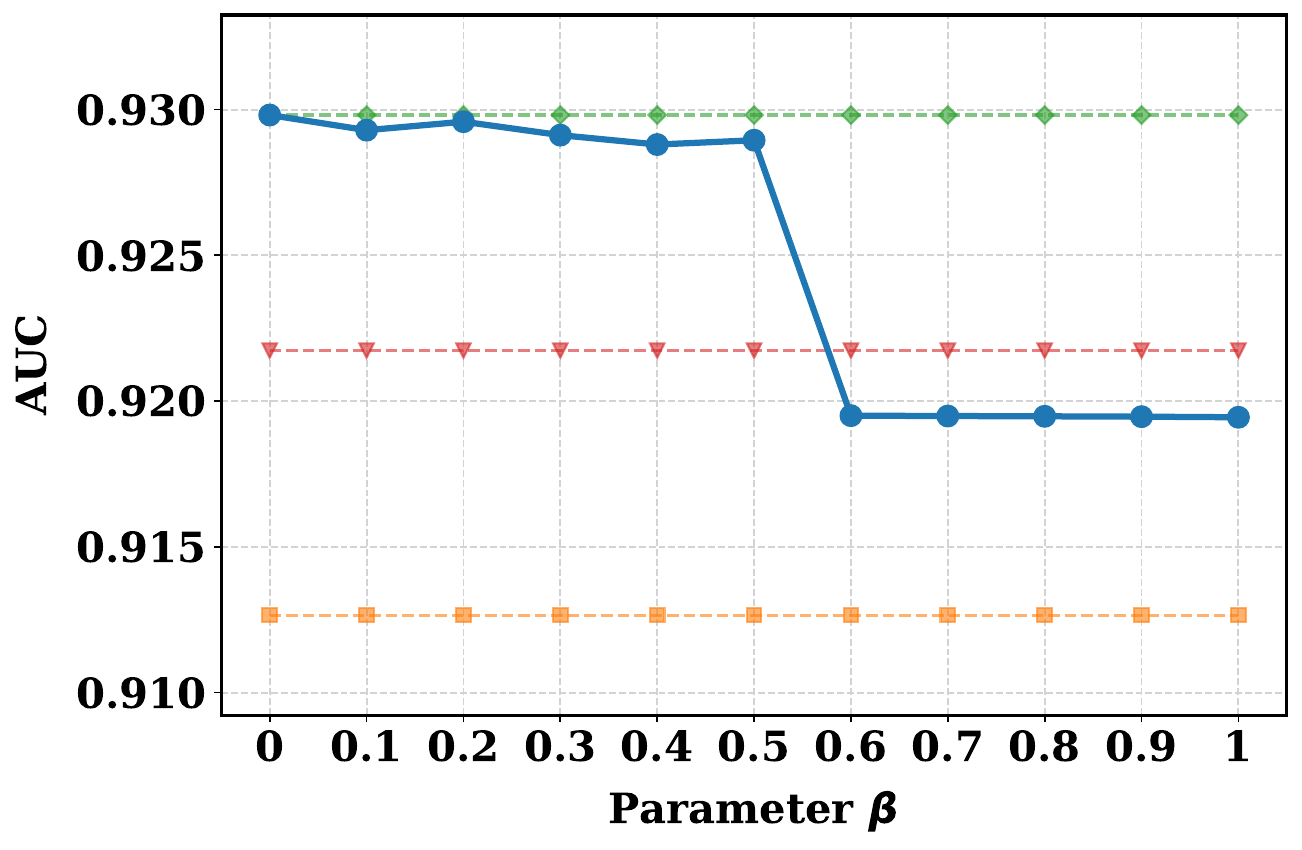}
    \end{subfigure}
    \caption{Sensitivity to the training cost coefficient $\beta$ on Adult Census. From left to right, the panels report training welfare, real welfare, utility, expected cost, cross-entropy loss, accuracy, Macro-F1, and AUC. The legend in the first panel applies to all panels.}
    \label{fig:beta_ablation}
\end{figure*}
Each benchmark dataset represents a separate buyer request, and no gate is shared across datasets. The gate is trained only on the request-specific training split, the validation split is used for model and baseline selection, and the test split is reserved for final evaluation. Thus, the experiments assess generalization to unseen instances from the same request distribution rather than zero-shot generalization to new tasks or datasets. All experiments are repeated using five random seeds. We report the mean and sample standard deviation of task-specific predictive metrics, expected service cost, and market welfare across these runs.

\subsection{Market Efficiency}
\label{market_effficiency}
Tables~\ref{tab:tabular_classification}, \ref{tab:tabular_regression}, and \ref{tab:image_classification} report the results for tabular classification,
tabular regression, and image classification. MoE Market achieves the
highest mean welfare on all 15 datasets and a lower mean expected cost
than Standard MoE in every case. Although Standard MoE is more accurate
on some datasets, MoE Market provides a more favorable balance between
predictive utility and execution cost. WSS incurs the lowest cost by
selecting a single expert, but MoE Market achieves higher welfare through
input-dependent coordination. Uniform Average generally performs worse
because it accounts for neither input-dependent expert capabilities nor
execution costs. Overall, cost-aware routing improves market efficiency
while retaining competitive predictive performance.

\subsection{Sensitivity Analysis of $\beta$}
\label{sensitivity of beta}

Figure~\ref{fig:beta_ablation} evaluates $\beta\in[0,1]$ in increments of $0.1$ using one representative Adult Census run. The independently fitted expert bank and measured cost profile remain fixed, while the MoE Market gate is retrained for each value of $\beta$. For a gate trained with coefficient $\beta$, Training Welfare is computed as $U_\beta-\beta C_\beta$. To compare all trained gates under the same market condition, we additionally report Real Welfare, defined as
\begin{equation}
{
\mathrm{Welfare}_{\mathrm{real}}(\beta)
=
U_\beta-\beta_{\mathrm{real}}C_\beta,
\qquad
\beta_{\mathrm{real}}=0.3.
}
\end{equation}

\begin{figure*}[htbp]
    \centering
    \begin{subfigure}{0.245\linewidth}\includegraphics[width=\linewidth]{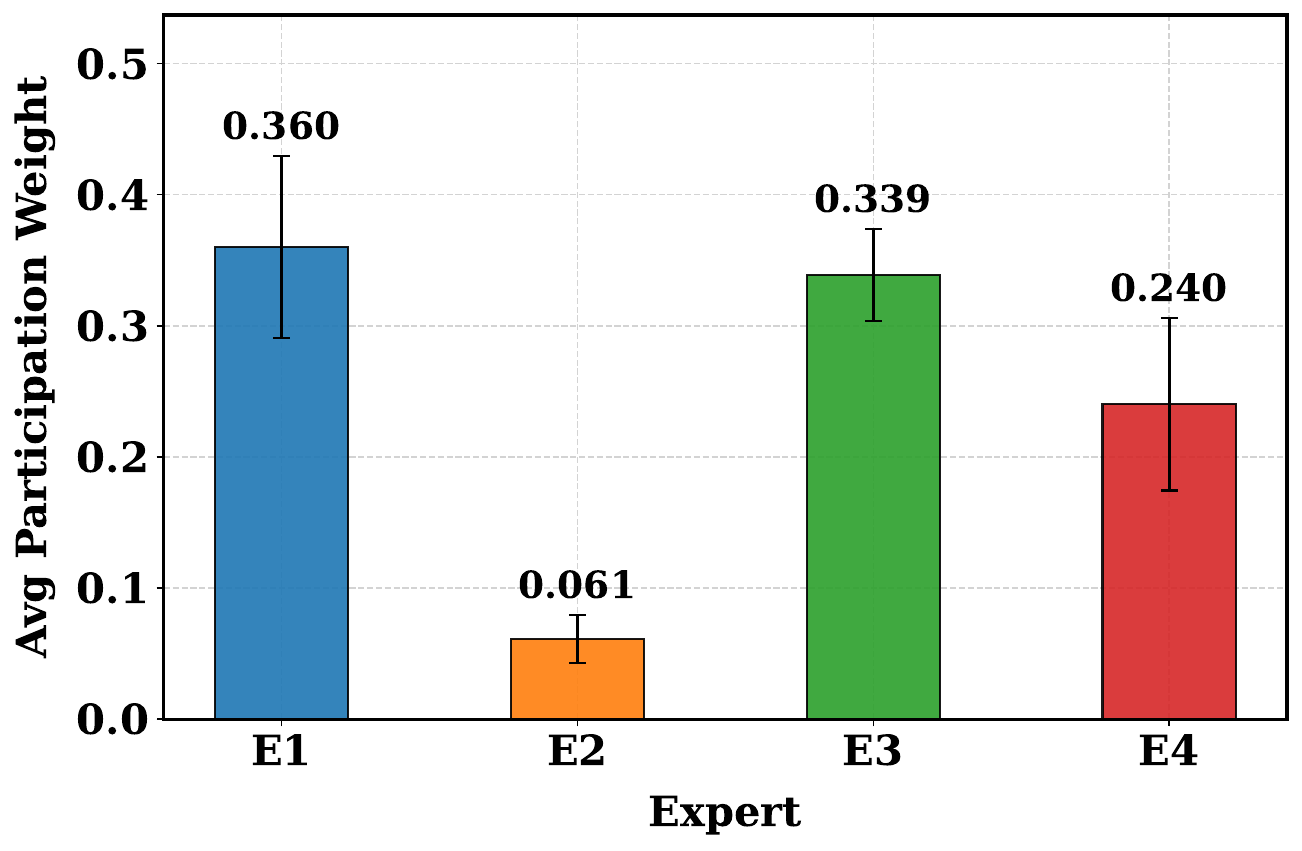}
    \end{subfigure}\hfill
    \begin{subfigure}{0.245\linewidth}\includegraphics[width=\linewidth]{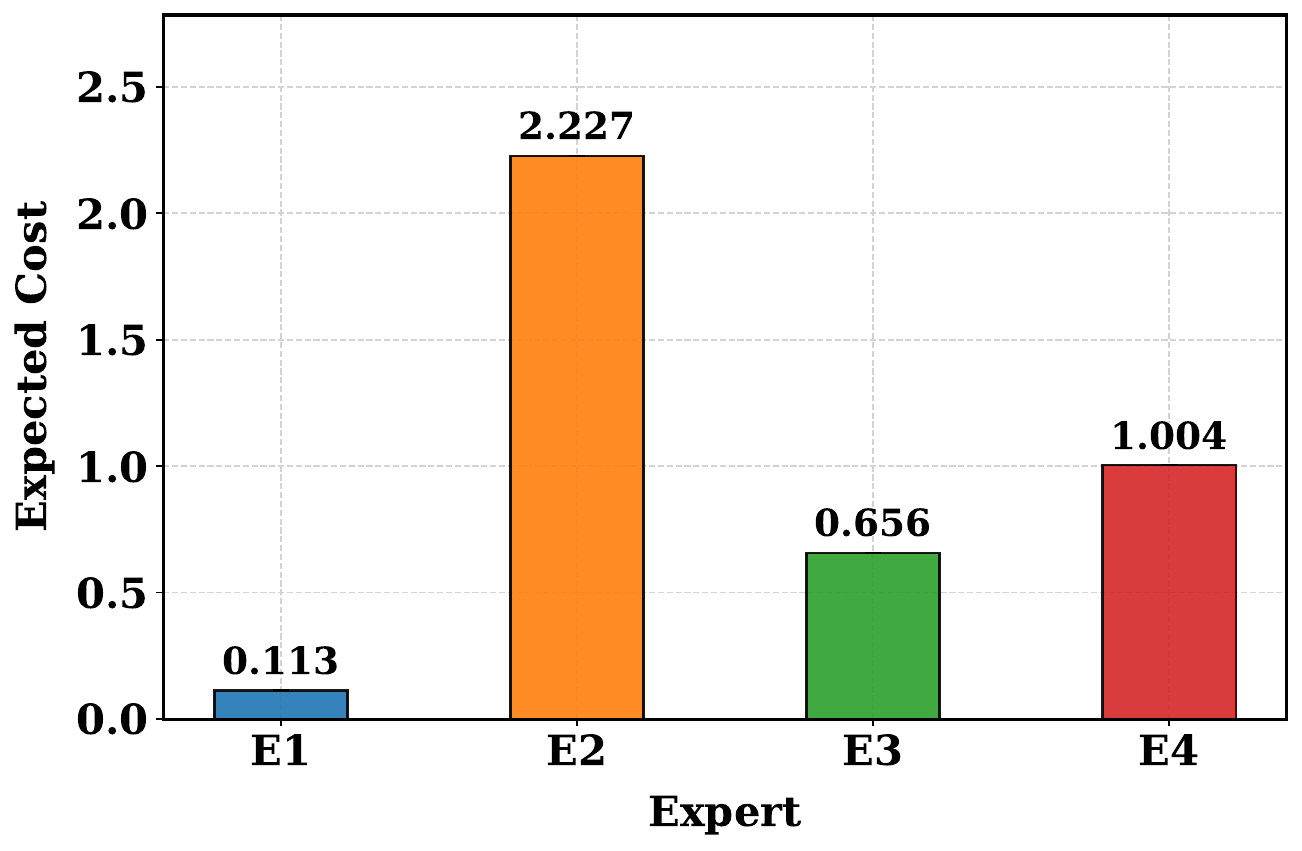}
    \end{subfigure}\hfill
    \begin{subfigure}{0.245\linewidth}\includegraphics[width=\linewidth]{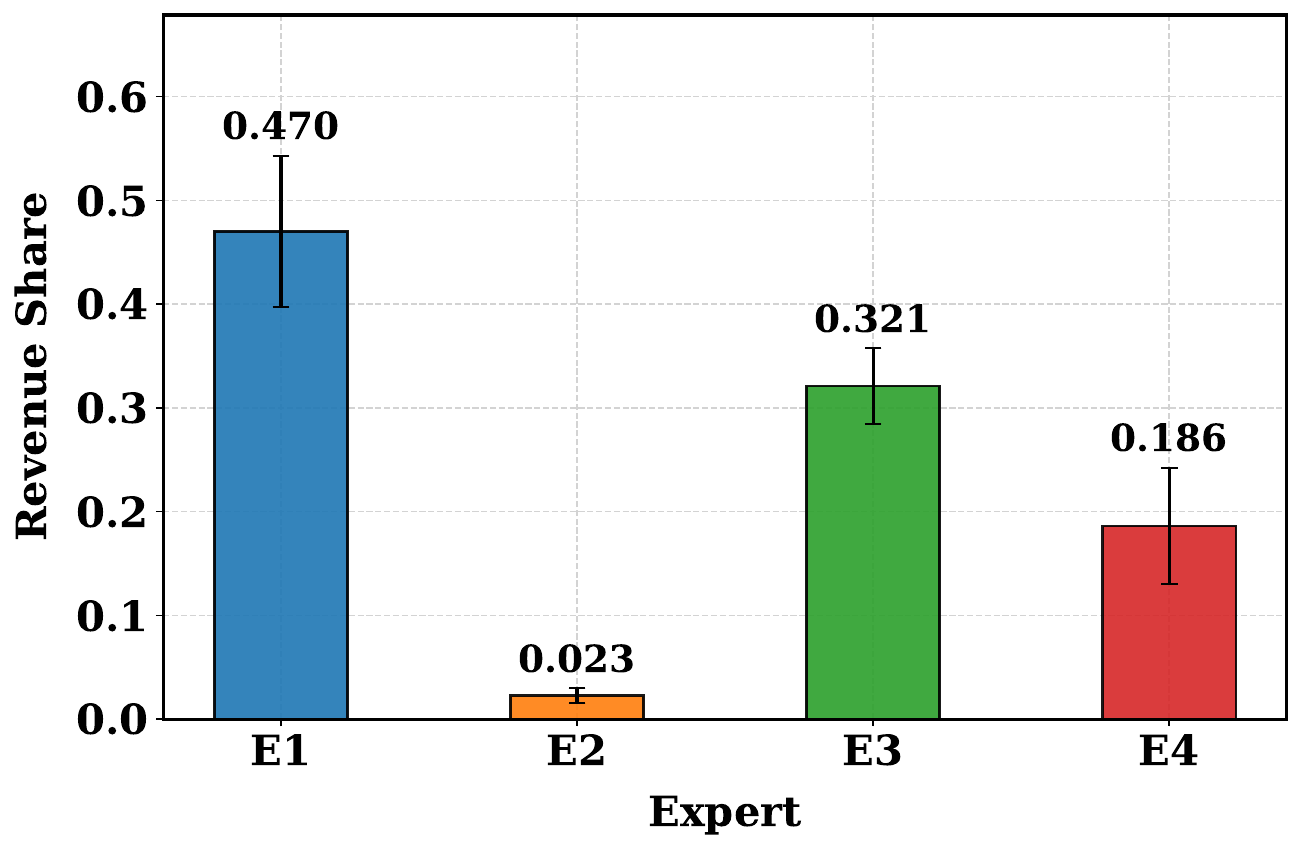}
    \end{subfigure}\hfill
    \begin{subfigure}{0.245\linewidth}\includegraphics[width=\linewidth]{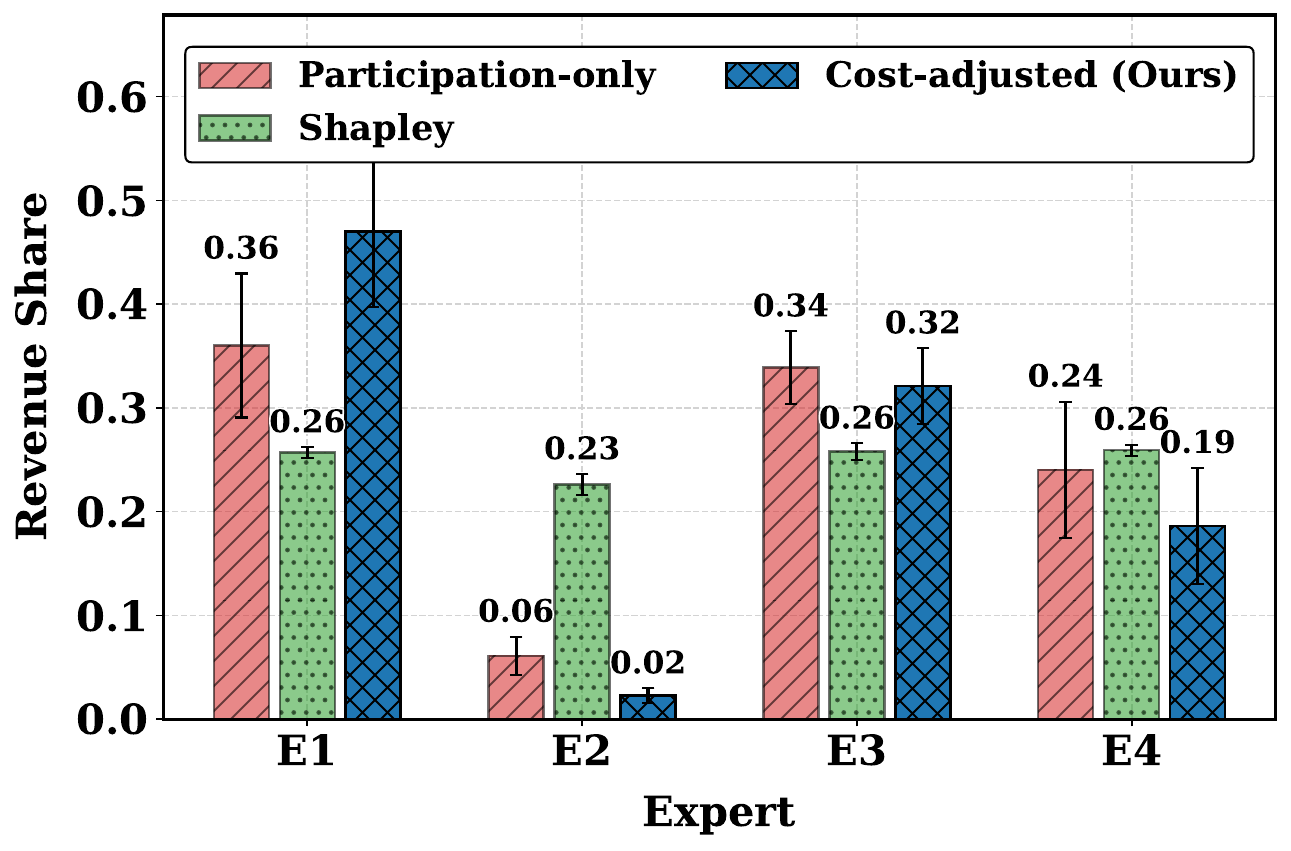}
    \end{subfigure}
    \caption{Revenue-allocation behavior on Breast Cancer (Case 1). From left to right, the panels show average participation weight, normalized expert unit cost, the proposed cost-adjusted revenue share at $\lambda=0.6$, and the comparison among participation-only, exact Shapley, and cost-adjusted allocation.}
    \label{fig:allocation1}
\end{figure*}

\begin{figure*}[htbp]
    \centering
    \begin{subfigure}{0.245\linewidth}\includegraphics[width=\linewidth]{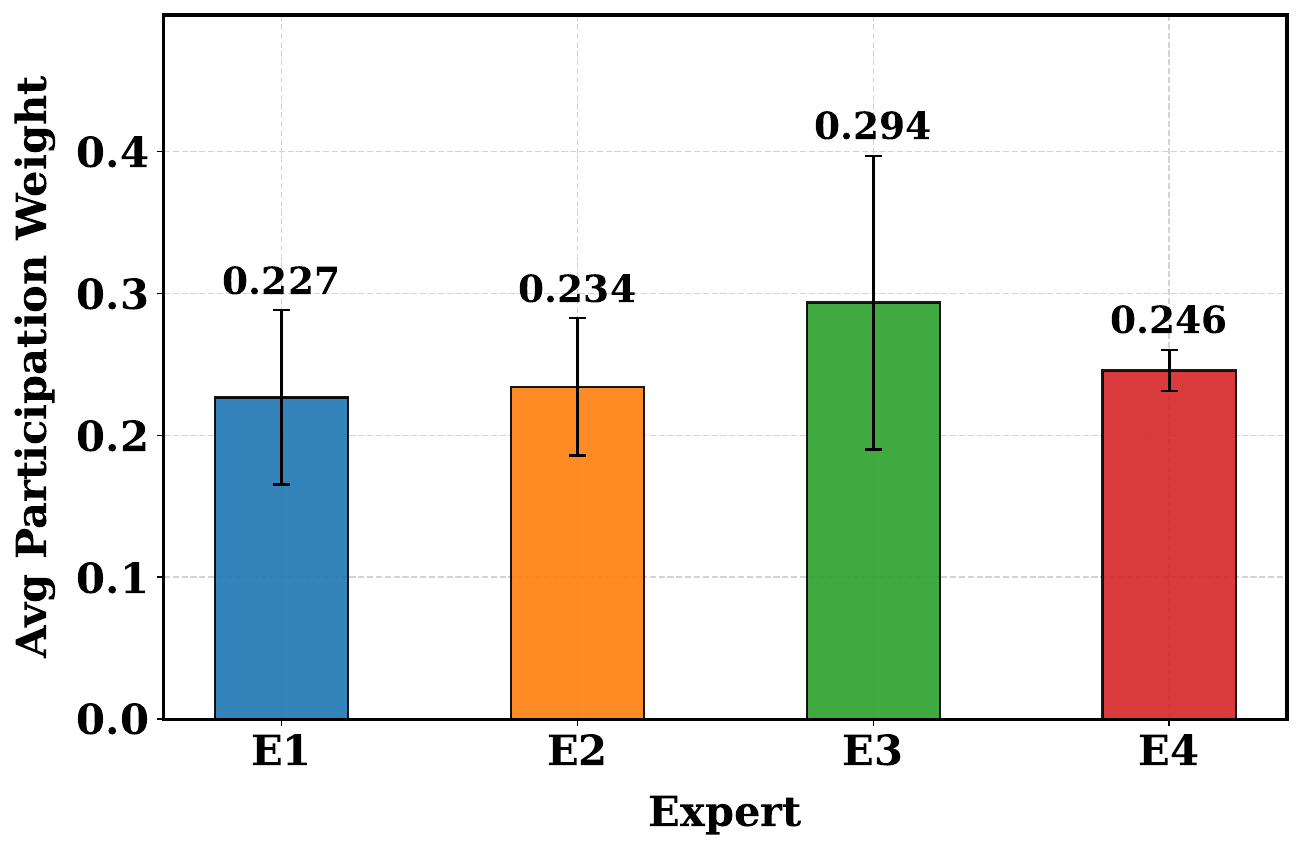}
    \end{subfigure}\hfill
    \begin{subfigure}{0.245\linewidth}\includegraphics[width=\linewidth]{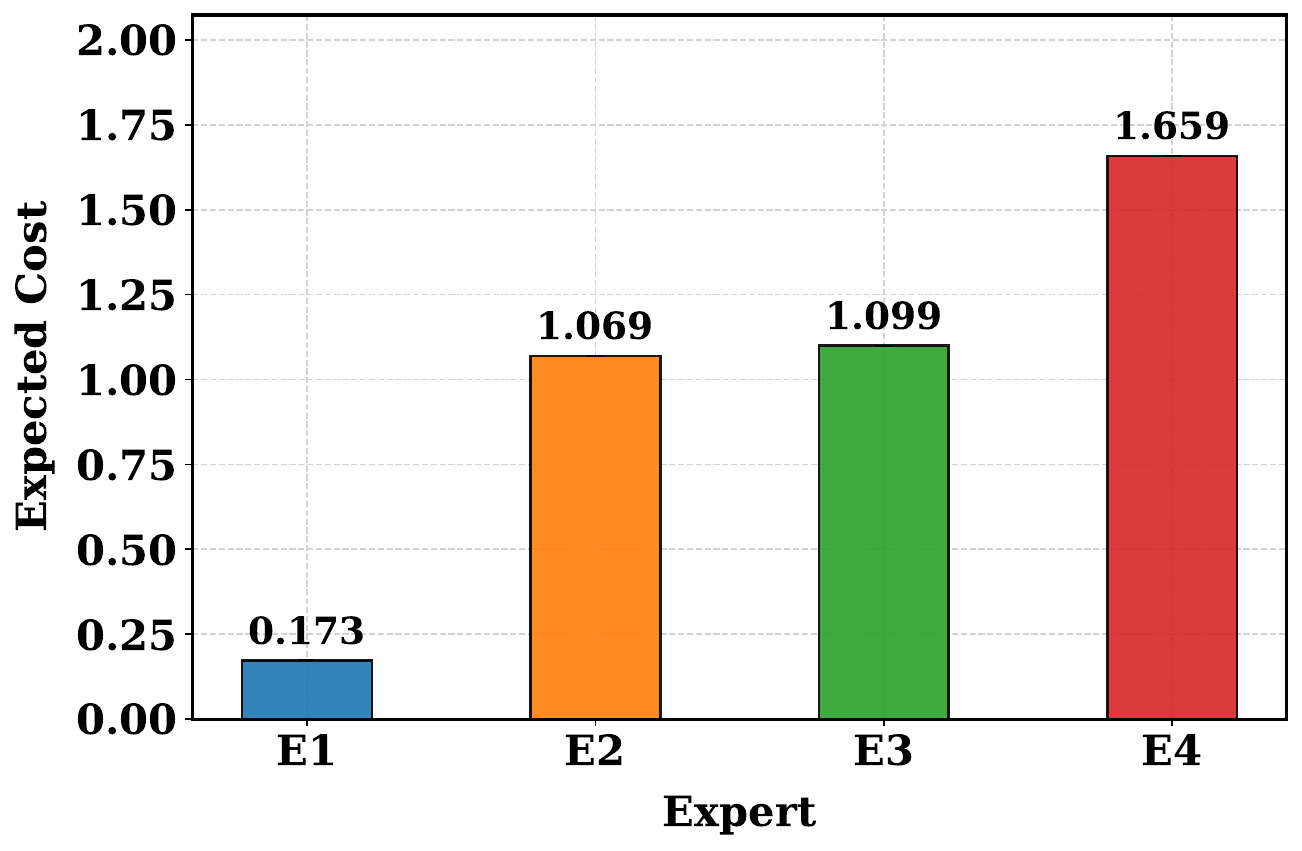}
    \end{subfigure}\hfill
    \begin{subfigure}{0.245\linewidth}\includegraphics[width=\linewidth]{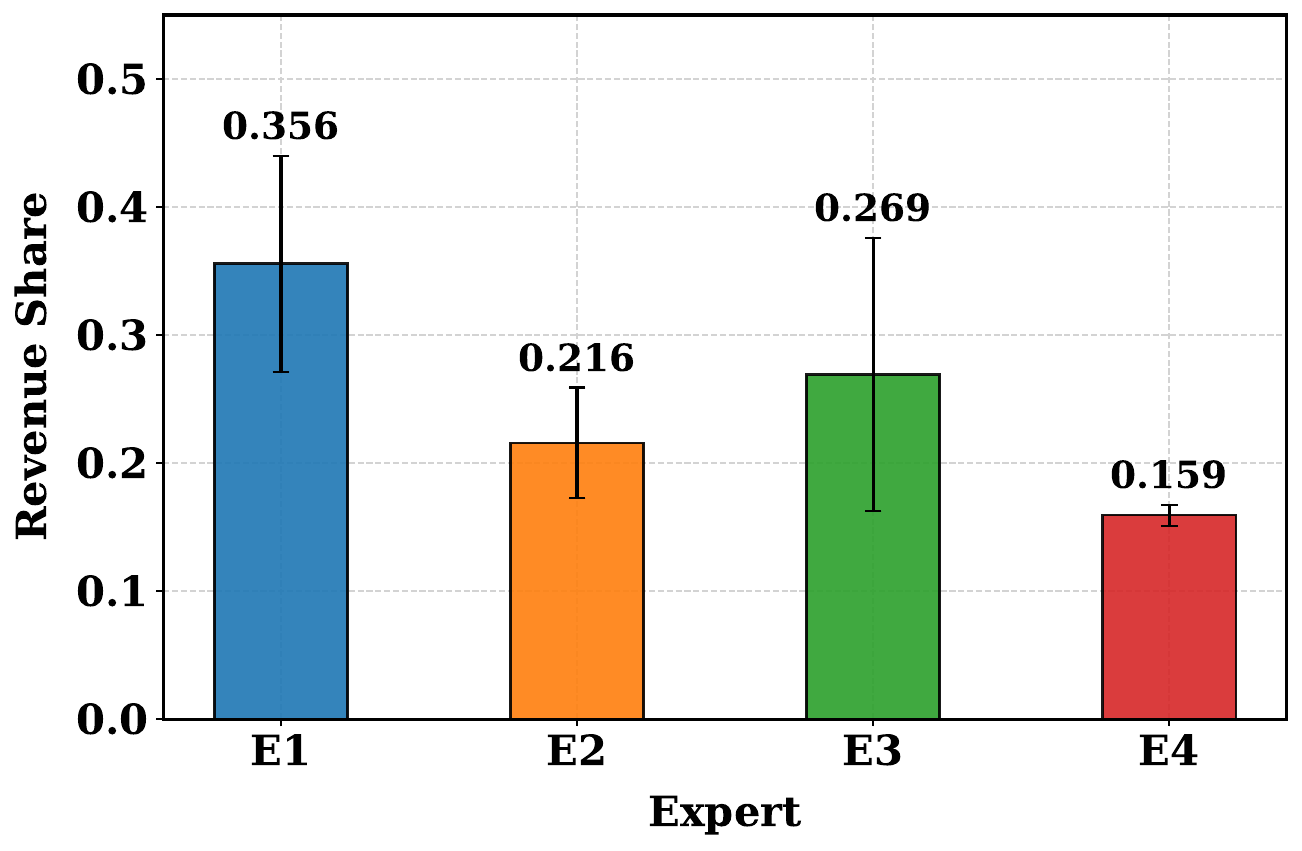}
    \end{subfigure}\hfill
    \begin{subfigure}{0.245\linewidth}\includegraphics[width=\linewidth]{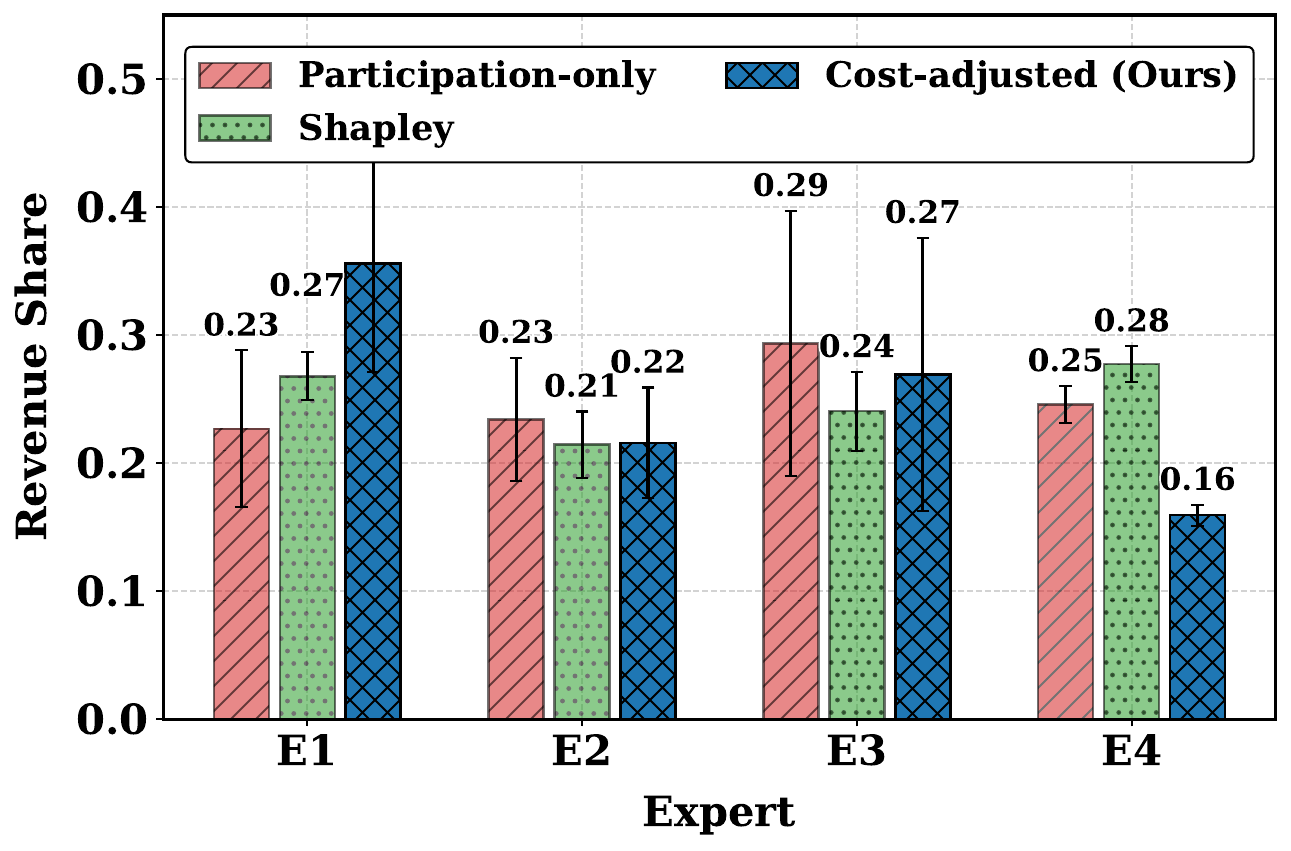}
    \end{subfigure}
    \caption{Revenue-allocation behavior on Diabetes (Case 2). From left to right, the panels show average participation weight, normalized expert unit cost, the proposed cost-adjusted revenue share at $\lambda=0.6$, and the comparison among participation-only, exact Shapley, and cost-adjusted allocation.}
    \label{fig:allocation2}
\end{figure*}

\begin{figure*}[htbp]
    \centering
    \begin{subfigure}{0.245\linewidth}\includegraphics[width=\linewidth]{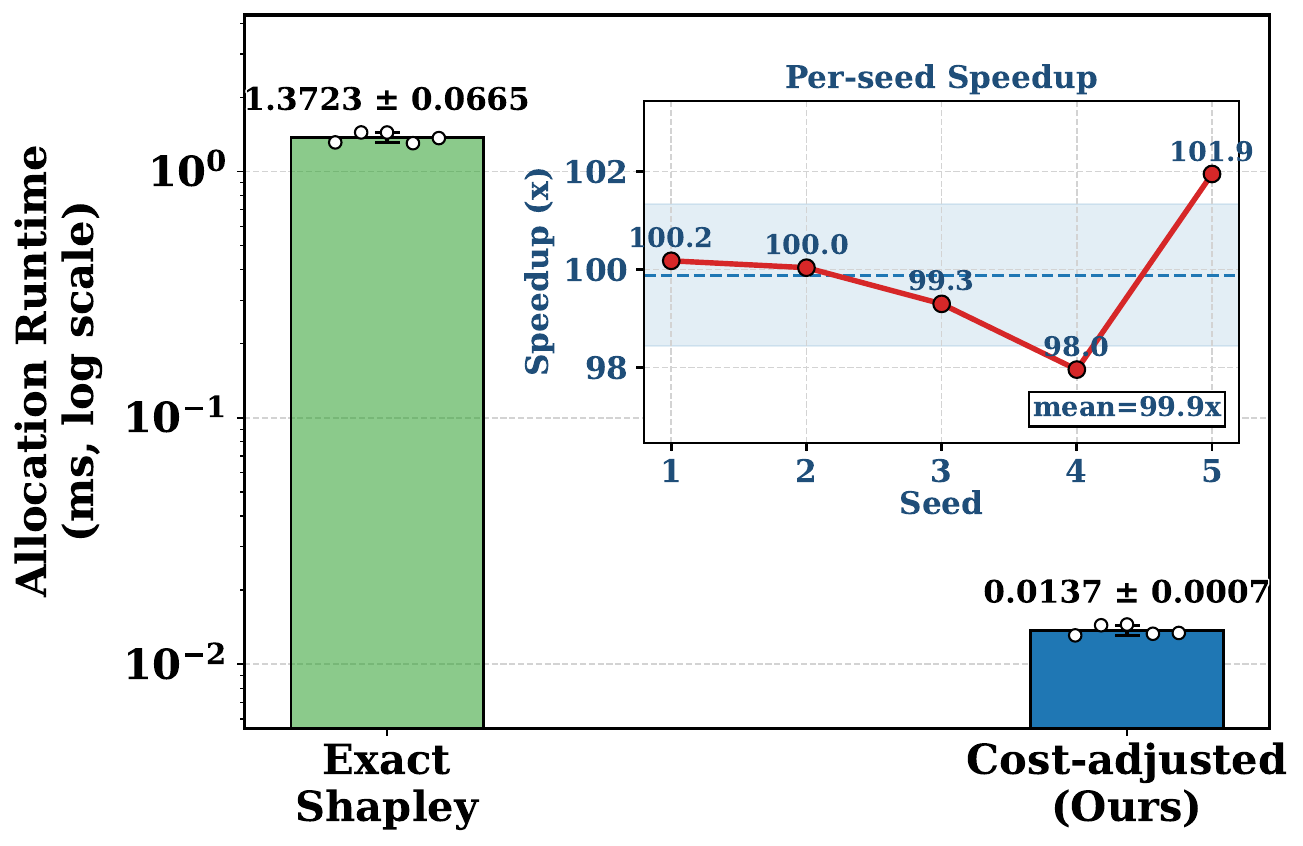}
    \end{subfigure}\hfill
    \begin{subfigure}{0.245\linewidth}\includegraphics[width=\linewidth]{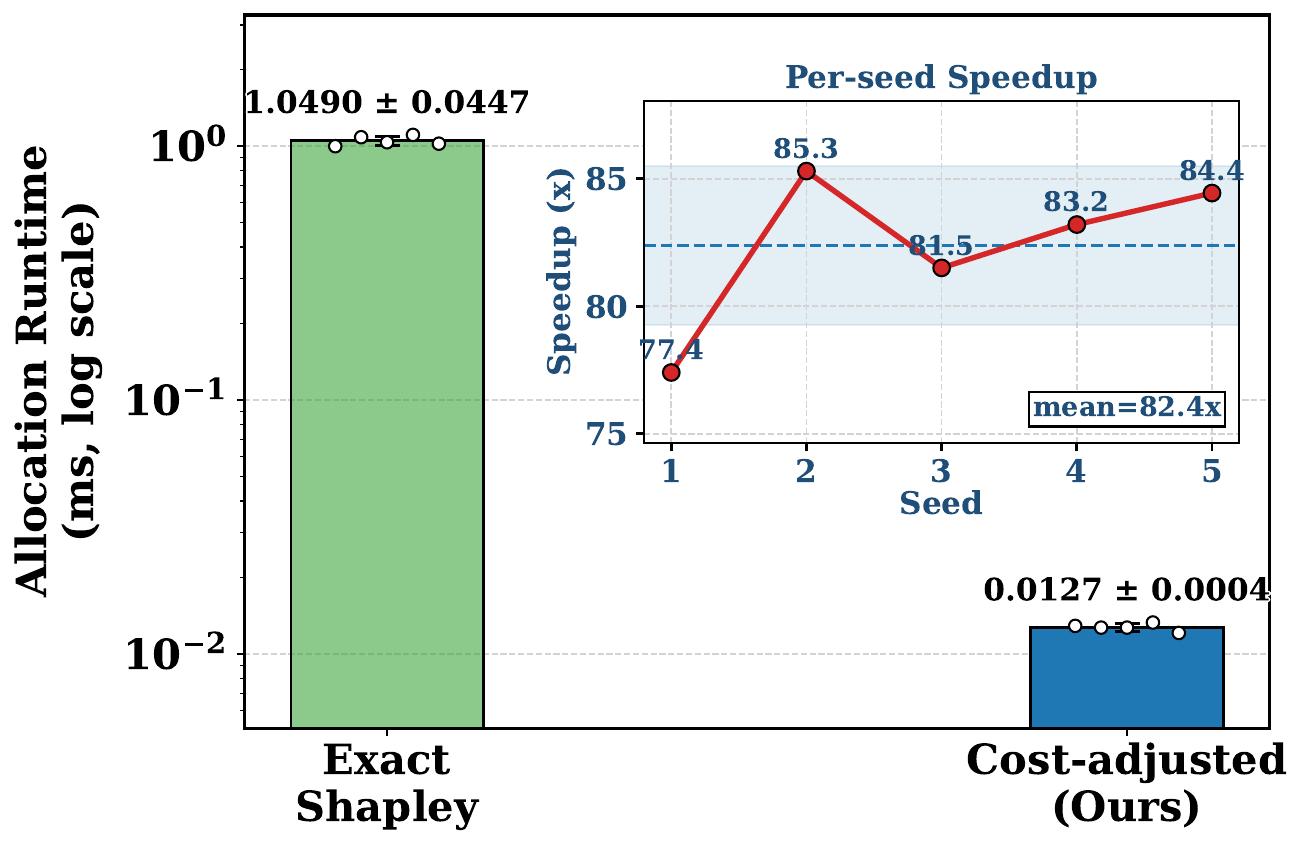}
    \end{subfigure}\hfill
    \begin{subfigure}{0.245\linewidth}\includegraphics[width=\linewidth]{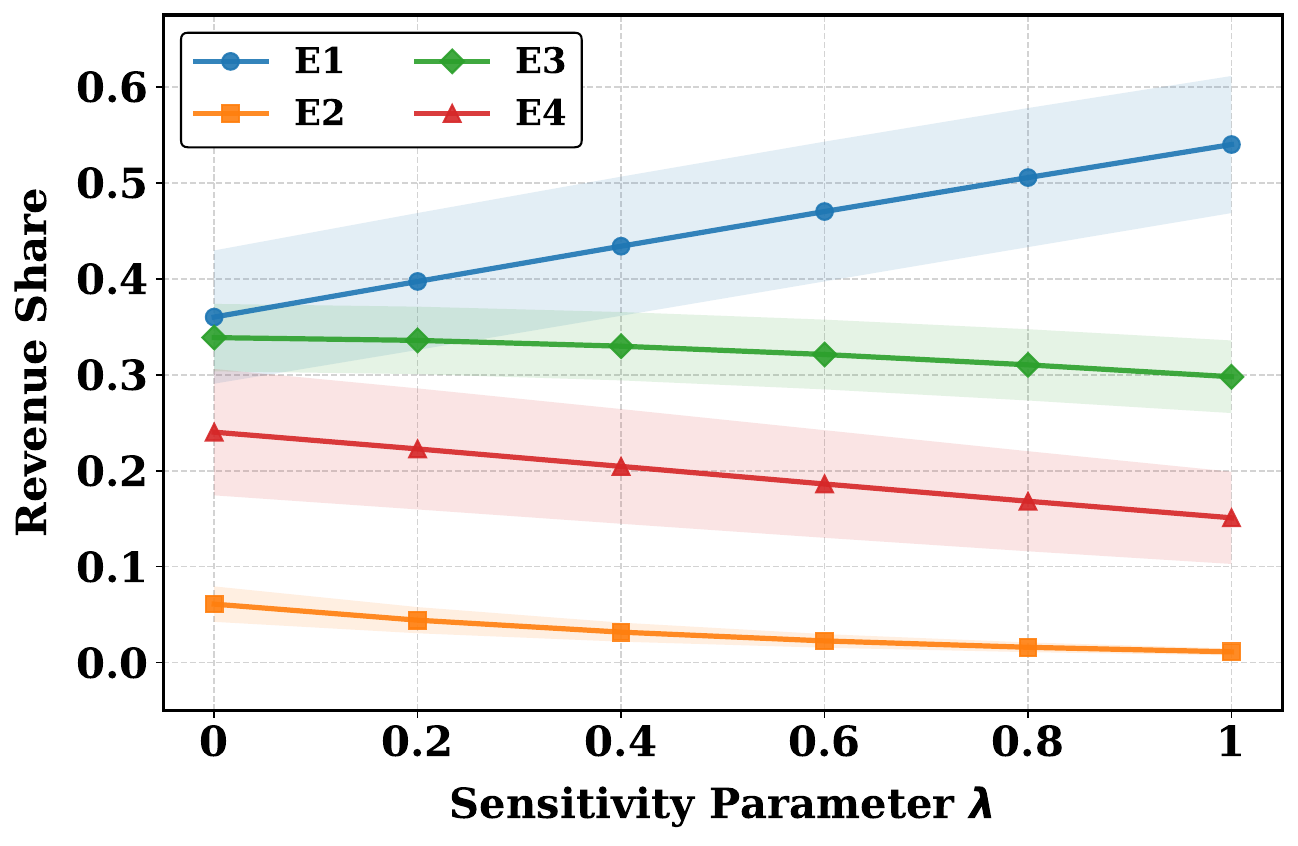}
    \end{subfigure}\hfill
    \begin{subfigure}{0.245\linewidth}\includegraphics[width=\linewidth]{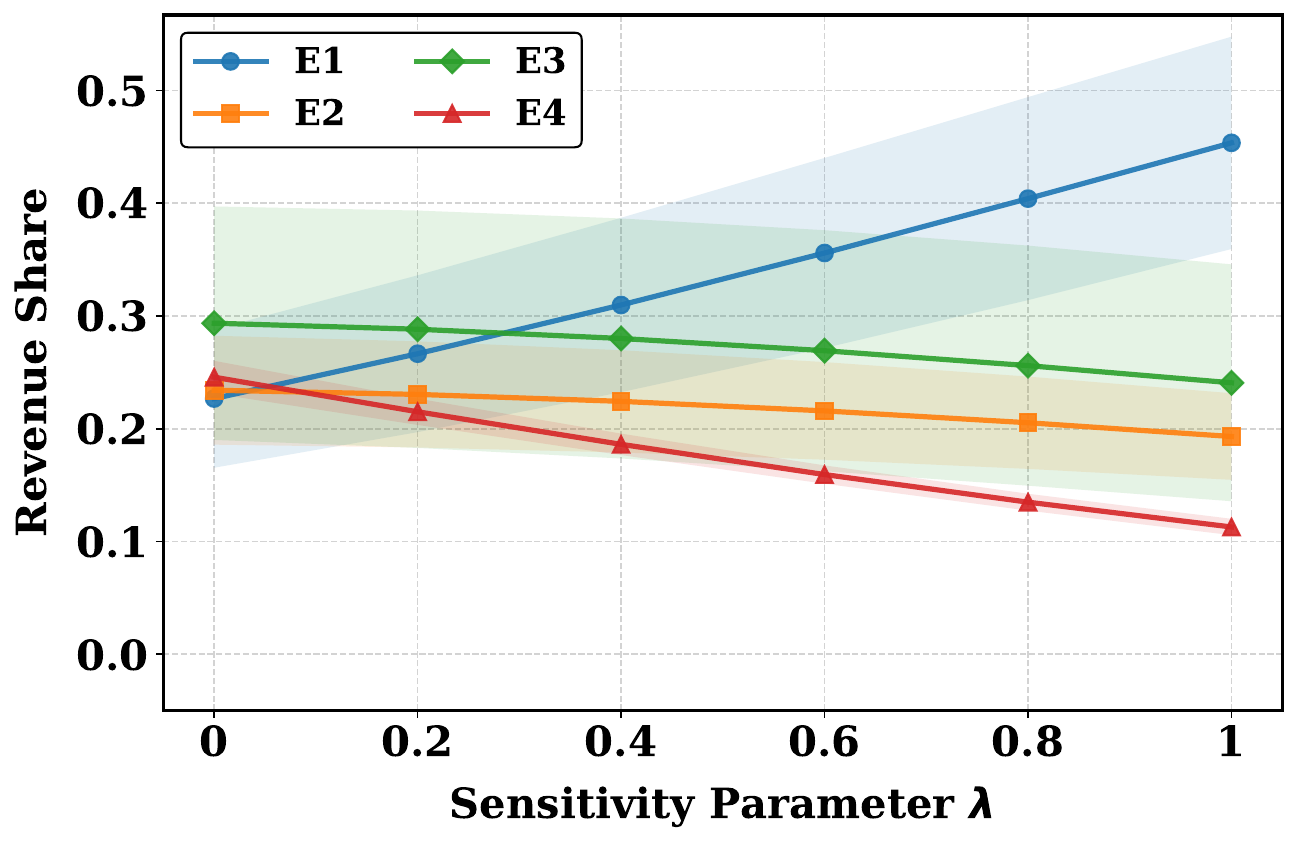}
    \end{subfigure}
    \caption{Runtime comparison and sensitivity to $\lambda$. From left to right, the panels report allocation runtime on Breast Cancer, allocation runtime on Diabetes, $\lambda$ sensitivity on Breast Cancer, and $\lambda$ sensitivity on Diabetes. Runtime bars show mean $\pm$ standard deviation over five seeds, and the inset panels report per-seed speedup over exact Shapley allocation.}
    \label{fig:allocation_runtime_lambda}
\end{figure*}

At $\beta=0$, MoE Market reduces to cost-unaware routing and closely matches Standard MoE. Increasing $\beta$ to $0.1$--$0.5$ reduces expected cost from approximately $0.43$ to $0.11$--$0.19$ with limited predictive degradation. Real Welfare peaks near $\beta=0.3$, where the training and deployment coefficients coincide. For $\beta\ge0.6$, further cost reductions no longer compensate for the decline in utility and predictive performance. Thus, $\beta$ controls the quality--cost trade-off, and alignment with the deployment condition avoids both cost-insensitive and excessively cost-oriented routing.

\subsection{Revenue Allocation}
\label{revenue_allocation}

We evaluate whether the proposed allocation rule reflects both realized expert participation and heterogeneous execution cost. We use Breast Cancer and Diabetes as two representative requests with different routing and cost patterns. All results are aggregated over five random seeds, and the default cost-discount coefficient is $\lambda=0.6$.

We compare the proposed rule with participation-only allocation and exact Shapley allocation. Participation-only allocation distributes revenue according to normalized average gating weights. Exact Shapley allocation follows the cooperative-game formulation of Shapley~\cite{shapley1953value}, which has also been used for contributor compensation in model marketplaces~\cite{liu2021dealer}. For each coalition $S\subseteq\mathcal{E}$, we define $v(S)$ as the predictive utility obtained by retaining only the experts in $S$ and renormalizing their original gating weights, with $v(\emptyset)=0$. We enumerate all $2^{|\mathcal{E}|}$ coalitions to obtain exact, rather than sampled, Shapley values. Negative values are clipped to zero and the remaining values are normalized to sum to one before allocating the residual revenue pool. Thus, this baseline captures marginal predictive contribution but does not explicitly account for measured execution cost.

Figures~\ref{fig:allocation1} and \ref{fig:allocation2} demonstrate the intended cost-sensitive adjustment. In Case~1, $E_1$ and $E_3$ have similar participation, but the lower-cost $E_1$ receives a larger share, while the costly and infrequently used $E_2$ receives little revenue. In Case~2, allocation similarly shifts toward the lower-cost $E_1$ and away from the higher-cost $E_4$. Unlike exact Shapley, which reflects marginal predictive utility, the proposed rule jointly accounts for realized routing and measured execution cost. The last two panels of Figure~\ref{fig:allocation_runtime_lambda} further show that increasing $\lambda$ consistently shifts revenue toward lower-cost experts; $\lambda=0.6$ provides a visible adjustment without concentrating all revenue on one expert.

The first two panels of Figure~6 show mean speedups of approximately $99.9\times$ on Breast Cancer and $82.4\times$ on Diabetes. Once participation statistics are available, the proposed rule is linear in the number of experts, whereas exact Shapley requires coalition enumeration.

\section{Conclusion}

This paper presents an MoE-based model market that coordinates heterogeneous experts as a market-level coordination mechanism for collaborative AI service delivery. 

We formalize the market participants, workflow, cost structure, and welfare objective, and develop a cost-aware gating mechanism that balances predictive utility against measured execution cost. We further introduce a cost-adjusted revenue allocation rule based on realized expert participation and cost, and establish its budget balance, participation monotonicity, and cost sensitivity. We evaluate the framework using independently trained and frozen MLP, Transformer, XGBoost, and CatBoost experts with latency-based execution costs. Five-seed experiments on 15 tabular and image datasets show that MoE Market achieves the highest mean welfare on every dataset and a lower mean expected cost than Standard MoE in every case, while retaining competitive predictive performance. Comparisons with Welfare-aware Static Selection demonstrate the benefit of input-dependent coordination over static single-model selection. The allocation experiments further show that the proposed rule produces cost-sensitive revenue shares and is substantially faster to compute than exact Shapley allocation.

Future work will focus on dynamic model markets in which expert availability, execution costs, and deployment distributions evolve over time.

\begin{acks}
This work was partially funded by the Horizon Europe projects DATAPACT (No. 101189771) and RAISE Suite (No. 101188337).
\end{acks}

\section*{Artifacts}
The artifacts for reproducing our experiments are available at: \url{https://github.com/MAYIZHOU/artifact-code_moe.git}.
The repository includes source code, experiment scripts, and reproduction instructions.


\balance
\bibliographystyle{ACM-Reference-Format}
\bibliography{ref}


\end{document}